\documentclass{article}
\usepackage{a4wide}

\usepackage[utf8]{inputenc} 
\usepackage[T1]{fontenc}    
\usepackage[hidelinks]{hyperref}      
\usepackage{url}            
\usepackage{booktabs}       
\usepackage{amsfonts}       
\usepackage{nicefrac}       
\usepackage{microtype}      
\usepackage{xcolor}         
\usepackage{tikz}
\usepackage{amsmath}
\usepackage{amssymb}
\usepackage{graphicx}%
\usepackage{color}

\usepackage{comment}

\newtheorem{theorem}{Theorem}

\newtheorem{claim}[theorem]{Claim}

\newtheorem{corollary}[theorem]{Corollary}

\newtheorem{lemma}[theorem]{Lemma}

\newtheorem{remark}[theorem]{Remark}

\newenvironment{proof}[1][Proof]{\textbf{#1.} }{\ \rule{0.5em}{0.5em}}
\newenvironment{claim*}{\textbf{Claim}}{}

\newcommand{\R}{\mathbb{R}}
\newcommand{\B}{\mathcal{B}}

\newcommand{\E}{\mathbb{E}}

\newcommand{\cost}{\mathrm{cost}^{(2)}}
\newcommand{\costa}{\mathrm{cost}^{(\alpha)}}
\newcommand{\OPT}{\texttt{OPT}}
\newcommand{\COPT}{\mathcal{C}_{\texttt{OPT}}}
\newcommand{\smax}{\sigma_{\mathrm{max}}}
\newcommand{\smin}{\sigma_{\mathrm{min}}}

\title{An Analysis of $D^\alpha$ seeding for $k$-means} 

\author{%
  Etienne Bamas\thanks{Post-Doctoral Fellow at the ETH AI Center, Switzerland.} 
   \and
   Sai Ganesh Nagarajan\thanks{EPFL, Switzerland.}
  \and
  Ola Svensson\thanks{EPFL, Switzerland.}
}

\date{}

\begin{document}

\maketitle

\begin{abstract}
One of the most popular clustering algorithms is the celebrated $D^\alpha$ seeding algorithm (also know as \texttt{$k$-means++} when $\alpha=2$) by Arthur and Vassilvitskii (2007), who showed that it guarantees in expectation an $O(2^{2\alpha}\cdot  \log k)$-approximate solution to the ($k$,$\alpha$)-means cost (where euclidean distances are raised to the power $\alpha$) for any $\alpha\ge 1$. More recently, Balcan, Dick, and White (2018) observed experimentally that using $D^\alpha$ seeding with $\alpha>2$ can lead to a better solution with respect to the standard $k$-means objective (i.e. the $(k,2)$-means cost).

In this paper, we provide a rigorous understanding of this phenomenon. For any $\alpha>2$, we show that $D^\alpha$ seeding guarantees in expectation an approximation factor of
\begin{equation*}
    O_\alpha \left((g_\alpha)^{2/\alpha}\cdot \left(\frac{\smax}{\smin}\right)^{2-4/\alpha}\cdot (\min\{\ell,\log k\})^{2/\alpha}\right)
\end{equation*}
with respect to the standard $k$-means cost of any underlying clustering; where $g_\alpha$ is a parameter capturing the concentration of the points in each cluster, $\smax$ and $\smin$ are the maximum and minimum standard deviation of the clusters around their means, and $\ell$ is the number of distinct mixing weights in the underlying clustering (after rounding them to the nearest power of $2$). For instance, if the underlying clustering is defined by a mixture of $k$ Gaussian distributions with equal cluster variance (up to a constant-factor), then our result implies that: (1) if there are a constant number of mixing weights, any constant $\alpha>2$ yields a constant-factor approximation; (2) if the mixing weights are arbitrary, any constant $\alpha>2$ yields an $O\left(\log^{2/\alpha}k\right)$-approximation, and $\alpha=\Theta(\log\log k)$ yields an $O(\log\log k)^3$-approximation. We complement these results by some lower bounds showing that the dependency on $g_\alpha$ and $\smax/\smin$ is tight.

Finally, we provide an experimental confirmation of the effects of the aforementioned parameters when using $D^\alpha$ seeding. Further, we corroborate the observation that $\alpha>2$ can indeed improve the $k$-means cost compared to $D^2$ seeding, and that this advantage remains even if we run Lloyd's algorithm after the seeding.
\end{abstract}

\section{Introduction}
Clustering is a quintessential machine learning problem with numerous practical applications in medicine \cite{alashwal2019application}, image segmentation \cite{shi2000normalized, burney2014k}, market analysis \cite{chiu2009intelligent} and anomaly detection \cite{munz2007traffic}, to name a few. One of the most popular formulations is the $k$-means problem that requires us to pick $k$ centers such that the sum of the squared distance from each data point to its closest center is minimized. The $k$-means problem is NP-hard even in 2 dimensions \cite{mahajan2009planar} and most research is therefore focused on heuristics and approximation algorithms. For a long time, a heavily used heuristic for this problem has been the Lloyd's algorithm, with Expectation Maximization (EM) style updates for the centers after an initial set of $k$ centers are chosen uniformly at random from the data. While this method finds a local optimum, it is known not to have any approximation guarantees and it could have an exponential run time in the worst case \cite{arthur2006slow}.

 \paragraph{The \texttt{$k$-means++} method.}  Arthur and Vassilvitskii \cite{arthur2007k} came up with the elegant \texttt{$k$-means++} method that carefully selects the initial centers (also called $D^2$ seeding) such that the next center is a data point that is chosen with probability that is proportional to its squared distance to its closest center, selected thus far (see \cite{ostrovsky2013effectiveness} for a concurrent work on a similar algorithm as well). This intuitively makes sense as this initialization is more likely to discover new clusters (that are far away) than simply selecting centers uniformly at random. Indeed, they proved that the initial choice of centers already forms a $O(\log k)$ approximation (in expectation) for this problem.
 They complemented this upper bound with a family of instances where the expected cost of the $D^2$ seeding is a factor $\Omega(\log k)$ times the optimum cost, showing that their analysis is tight.

 \paragraph{Limitations on clusterable instances.} While the $D^2$ seeding method provides clear improvements over uniformly at random initialization in an elegant and efficient manner, the family of instances that show  the tightness of the  analysis indicates some of its limitations. The family of instances presented in~\cite{arthur2007k}  is indeed highly clusterable:  $k$ regular simplices of radius one (each of $n/k$ points)  and the pairwise distance between the centers of two simplices are $\Delta$. 
 As $\Delta$ tends to $\infty$ (i.e., the instance becomes more and more clusterable), the expected approximation guarantee of the $D^2$ seeding method tends to $\Theta(\log k)$. 

For such clusterable instances, the issue is  that the $D^2$ seeding method does not put enough probability mass on discovering new clusters. This phenomenon was already observed in the original paper (\cite{arthur2007k}), and they proposed a greedy variant that takes several samples at each iteration (increasing the probability that at least one hits a new yet undiscovered cluster) and makes a greedy choice among them. 
This greedy variant has worse guarantees in the worst case (see \cite{GrunauORT23} for a recent nearly tight analysis). However, the greedy variant shows better experimental performance (after all, we usually look for $k$ clusters when the data is clusterable), and is currently the method implemented in the popular Scikit-learn library \cite{pedregosa2011scikit}. Specifically, at each iteration $2+\log(k)$ points are sampled, and, among them, the point that decreases the objective the most is greedily chosen.

\paragraph{Data-driven approach.} More recently, Balcan et al. \cite{balcan2018data} proposed a data-driven approach in order to address the aforementioned limitation of the $D^2$ seeding method on clusterable data. Instead of always using $D^2$ seeding for the initial centers, they proposed to use $D^\alpha$ seeding where $\alpha$ is now a parameter of the algorithm. In $D^\alpha$ seeding,  a point is selected as the next center with probability proportional to its $\alpha$-powered distance to its closest center, selected thus far\footnote{We remark that $D^\alpha$ seeding was already considered in \cite{arthur2007k} but they studied it on a cost that was proportional to the distances raised to the power $\alpha$ (i.e. the $(k,\alpha)$-means cost) instead of the standard $k$-means objective. They showed that $D^{\alpha}$ seeding was $O(2^{2\alpha} \log k)$-approximate on this $\alpha$-cost. The use of $D^{\alpha}$ on other objectives, including the standard squared distance objective was first introduced in \cite{balcan2018data}, and they considered a data-driven approach.}.  One can observe that a large choice of $\alpha > 2$ increases the probability that a sampled center will discover a new yet undiscovered cluster which is advantageous. At the same time, a large $\alpha$ makes the algorithm more sensitive to outliers (which is also the reason why the greedy variant of \texttt{$k$-means++} is worse in the worst case). Hence the ``optimal'' selection of $\alpha$ depends on the kind of instances one wants to solve, which motivates a data-driven approach. One of the main results in~\cite{balcan2018data} is that this is feasible.  
They showed that the parameter $\alpha$ is learnable in the sense that if we assume the instance is drawn from some unknown distribution $\mathcal D$,  then with only polynomially many (in the instance size and other relevant parameters) samples and a polynomial running time, one can compute a parameter $\tilde \alpha$ for the sampling that is almost optimum for the given distribution $\mathcal D$. 
This is especially interesting as it shows that setting $\alpha$ to a good value on a given distribution is in principle a task that is manageable. 
Additionally, Balcan et al. \cite{balcan2018data} complemented their theoretical results with an experimental analysis that shows that setting $\alpha$ equal to $2$ is not always the best choice. 
For instance, on the MNIST dataset, they find that setting $\alpha$ close to $4$ is a significantly better choice than $\alpha=2$. This is even more striking in the case where $\mathcal D$ is a mixture of Gaussians, in which case setting $\alpha$ close to $20$ seems the best choice. 
This highlights the fact that in practice, one can outperform the popular \texttt{$k$-means++} algorithm of \cite{arthur2007k} by tweaking the parameter $\alpha$. 
Yet, they do not provide any quantitative understanding of this phenomenon nor provide any approximation guarantees on these instances with different $\alpha$. This is the main focus of this paper.

\paragraph{Our contributions.} 
Our main contribution is a rigorous analysis of the advantage of $D^\alpha$ seeding,  proving that it leads to constant-factor approximation guarantees for a large class of instances, including a balanced mixture of $k$ Gaussians, where the standard \texttt{$k$-means++} algorithm is already no better than $\Omega(\log k)$ (see Section~\ref{sec:thmproof}). We remark that a beyond worst-case analysis is essential as it is easy to see that $\alpha = 2$ is an optimal choice in the worst-case (just as the greedy variant of $k$-means++ is worse in the worst-case). 

In our beyond worst-case analysis, we identify natural data-dependent parameters that measure (i) how concentrated points are in clusters (the parameter $g_\alpha$), (ii) the ratio of the maximum and minimum standard deviation of the optimal clusters around their mean ($\frac{\smax}{\smin}$), and (iii) how balanced clusters are in terms of number of points (the parameter $\ell$). Using these parameters, we show that $D^\alpha$ seeding guarantees for any $\alpha>2$ an 
\begin{equation*}
    O_\alpha \left((g_\alpha)^{2/\alpha}\cdot \left(\frac{\smax}{\smin}\right)^{2-4/\alpha}\cdot (\min\{\ell,\log k\})^{2/\alpha}\right)
\end{equation*} approximation with respect to the standard $k$-means objective (the formal statement can be found in Section \ref{sec:result_statement}). We further show that the dependence on the first two parameters is necessary and tight (formal statement in Section~\ref{sec:result_statement}), and this dependence gives a theoretical explanation of the importance of selecting $\alpha$ as a function of the data (Section~\ref{sec:choosing_alpha}).  We leave it as an interesting open problem to understand the necessity of the third parameter $\ell$.  

Finally, a more open-ended direction following our work is to give a beyond-worst-case analysis of the greedy variant of \texttt{$k$-means++}. We take a first step in this direction by proving a negative result: we give a family of instances where the natural parameters (i)-(iii) are all constant (and thus $D^\alpha$ seeding yields a constant-factor approximation guarantee for any constant $\alpha>2$) but greedy \texttt{$k$-means++} as implemented in the Scikit-learn library (with $\Theta(\log k)$ samples per iteration) has a super constant approximation guarantee (see Section~\ref{sec:result_statement}).

\subsection{Further related works}

Several variants have been studied since the original publication of the \texttt{$k$-means++} method in~\cite{arthur2007k, ostrovsky2013effectiveness}. Aggarwal et al. \cite{aggarwal2009adaptive} obtained an $O(1)$-approximation with constant probability by selecting $O(k)$ centers, which was improved later by \cite{wei2016constant, makarychev2020improved}. Bahmani et al. \cite{bahmani2012scalable} provided a scalable version $k$-means++ that is even more practical, and \cite{BachemL017,Cohen-AddadLNSS20} provided faster ways for randomly selecting the centers. Recently, Lattanzi and Sohler \cite{lattanzi2019better} obtained a $O(1)$-approximation with additional $O(k\log\log k)$ steps of local search after using $k$-means++ to choose the initial centers and this was improved by Choo et al. \cite{choo2020k} to obtain a $10^{30}$-approximation with $\epsilon k$ steps of local search. Our constant factor guarantees are arguably smaller without any additional steps of local search but are applicable to the appropriate family of instances, whilst their method is applicable in the worst case across all instances. Moreover, their local search methods can be augmented on top of our guarantees of $D^{\alpha}$-seeding, which may offer significant improvements. Tangentially, one could study algorithms for learning the cluster centers when the data is instantiated from a mixture of Gaussians \cite{dasgupta1999learning,arora2005learning}. Furthermore, specific clustering algorithms are created under specific assumptions on the instances with various clusterability notions \cite{ackerman2009clusterability}. However, these clusterability notions are often (computationally) hard to check and the algorithms are not as efficient and simple as the seeding-based algorithms, which also work well in practice (without any assumptions). Finally, the main inspiration of this paper is from the idea of data-driven clustering by Balcan et al. \cite{balcan2018data}.

\subsection{Preliminaries and notations}
To formally introduce the $k$-means problem and the seeding algorithms, we will need to work with a metric space $(\mathbb{R}^d,\lVert \cdot \rVert_2)$. Since we always work with the Euclidean norm, we will drop the subscript in the notation and write the Euclidean norm of a vector $x$ to be $\lVert x \rVert$. If we are given some data points $\mathcal{X} \subset \mathbb{R}^d$, the cost of our data $\mathcal{X}$ associated with a given set of $t$ centers $Z_t$ can be defined as follows:
\begin{align}
    \cost(\mathcal{X},Z_t):=\sum_{x \in \mathcal{X}}\min_{c \in Z_t}\lVert x-c \rVert^2.
\end{align}
Now note that the centers $Z_t$ define a natural partition of the data, in that: 
$C_j = \{x \in \mathcal{X}: c_j=\arg \min_{c \in Z_t}\lVert x-c\rVert^2 \}$. Then one can write the cost equivalently in the following useful way:
\begin{align}
    \cost(\mathcal{X},Z_t)=\sum_{j=1}^{t}\sum_{x \in C_j}\lVert x-c_j \rVert^2.  
\end{align}
Furthermore, if one has a candidate clustering $\mathcal{C}$, then each corresponding center can be computed as $c_j=\frac{1}{|C_j|}\sum_{x \in C_j}x$, which are the centroids of the corresponding clusters. We will denote the optimal centroids by $\{\mu_1,\mu_2,\ldots,\mu_k\}$. By a slight abuse of notations, we might drop the subscript to identify a cluster $C\in \mathcal C$, and $\mu_C$ will refer to the mean of that cluster. Let $\COPT$ be the optimal clustering whose corresponding centers are $Z_{\OPT}$. By definition,
\begin{align}
    \cost(\mathcal{X},Z_{\OPT})=\min_{Z \subset \mathbb{R}^d: |Z|=k} \cost(\mathcal{X},Z).
\end{align}

\paragraph{The \texttt{$k$-means++} algorithm.} We will now describe the class of parameterized seeding algorithms as in \cite{arthur2007k,balcan2018data} for the $k$-means objective. For any $\alpha \in [0,\infty]$, the general $D^{\alpha}$ seeding procedure chooses $k$ centers as follows:

\begin{enumerate}
    \item The first center $z_1 \in \mathcal{X}$ is chosen uniformly at random from the data points.
    \item Let $Z_t$ be the set of $t$ centers chosen so far, such that $t < k$. The next center, $z_{t+1} \in \mathcal{X}$ is chosen with probability given by:
    \begin{align}\label{eqn:palpha}
        p^{(\alpha)}_{\mathcal{X}}(z):=\mathbb P(z \;\text{is chosen to be the next center}| Z_t ) = \dfrac{\min_{c \in Z_t} \lVert z-c \rVert^{\alpha}}{\sum_{z \in \mathcal{X}}\min_{c \in Z_t} \lVert z-c \rVert^{\alpha}}
    \end{align}
\end{enumerate}

The classic \texttt{$k$-means++} algorithm is the special case of $D^2$ seeding. For any $\alpha \ge  1$, Arthur and Vassilvitskii \cite{arthur2007k} show that  $D^{\alpha}$-seeding procedure is an $O(2^{2\alpha} \log k)$ approximation in expectation if the cost function is given by:
\begin{align}
    \costa(\mathcal{X},Z_k):= \sum_{x \in \mathcal{X}}\min_{c \in Z_k}\lVert x-c \rVert^{\alpha}\ .
\end{align}

Our focus in this paper is to provide guarantees on the $D^{\alpha}$ seeding algorithm with $\alpha > 2$ for the standard $k$-means objective (i.e. $\alpha=2$).

\paragraph{The greedy \texttt{$k$-means++} algorithm.} Although this is not the main focus of the paper, it is helpful to define briefly here the greedy variant of \texttt{$k$-means++}. The greedy variant with $m$ samples works as follows.
\begin{enumerate}
    \item The first center $z_1 \in \mathcal{X}$ is chosen uniformly at random from the data points.
    \item Let $Z_t$ be the set of $t$ centers chosen so far, such that $t < k$. We select a set of $m$ candidate centers $z_1,z_2,\ldots, z_m$ where each candidate is sampled according to the probability distribution 
    \begin{align}
        p_{\mathcal{X}}(z):=\mathbb P(z \;\text{is sampled}| Z_t ) = \dfrac{\min_{c \in Z_t} \lVert z-c \rVert^{2}}{\sum_{z \in \mathcal{X}}\min_{c \in Z_t} \lVert z-c \rVert^{2}}
    \end{align}
    The next added center $z_{t+1}$ is selected to be the one which decreases the cost the most, among all the candidate centers $z_1,z_2,\ldots, z_m$. 
    
    \end{enumerate}

Usually $m$ is selected to mildly increase with the input size. For instance the standard scikit-learn library implements the greedy version of \texttt{$k$-means++} using $m=\Theta (\log k)$ candidates at each step (see \cite{pedregosa2011scikit}).

\section{Our Results}\label{sec:result_statement}
In this section, we define formally the natural parameters that $D^{\alpha}$ seeding depends on and state formally our results. Moreover, we will provide a short discussion on the necessity of the dependence that will clarify our claims. The last part of this section focuses on using our results to provide recommendations on choosing $\alpha$ in different scenarios.

Before we move on to stating our results on the $D^{\alpha}$ seeding, we need to define the following quantities with respect to the optimal\footnote{Although we state our definitions and results with respect to the optimal clusters, our results hold for any reference clustering that satisfies the aforementioned properties.} clustering $\COPT=\{C_1,C_2,\ldots,C_k\}$.

\begin{enumerate}
    \item We define $\sigma_C$ as the standard deviation of the points inside cluster $C\in \COPT$. More precisely,
    \begin{equation}
        \sigma_C := \sqrt{\frac{\sum_{x \in C}\lVert x-\mu_C\rVert^2}{|C|}}\ .
    \end{equation}
    Following this, $\smax$ is defined as the maximum standard deviation of points inside a given cluster, i.e. $\smax:=\max_{C\in \COPT}\sigma_C$, and similarly $\smin:=\min_{C\in \COPT}\sigma_C$. 
    \item We need a parameter $g_\alpha$ that is  a ``moment'' condition that measures the concentration of the distances of the points to the centroid $\mu_C$ in a cluster $C$: 
\begin{equation}
\label{eq:concentration_assumption_new}
    g_\alpha:=\max_{C\in \COPT} \frac{(1/|C|^2)\cdot \sum_{z\in C} \costa(C,z)}{(\cost (C,\mu_C)/|C|)^{\alpha/2}} \ .
\end{equation}
We will give more intuition on the parameter $g_\alpha$ in Section \ref{sec:param}.
\item Finally, we need a parameter $\ell$ to control the number of distinct weights of clusters (where the weight of a cluster $C$ is simply equal to $|C|$). Formally, for any integer $i\ge 0$ we let $k_i$ to be the number of clusters of $\COPT$ whose weights lie in the interval $[2^{i},2^{i+1})$. Then we define the following key parameter.
\begin{equation}
\label{eqn:ell_parameter}
    \ell := |\{i\ge 0 \mid k_i>0\}|\ .
\end{equation}
\end{enumerate}

\begin{remark}\label{rem:OPT}
    Note that we can express $\OPT(C)$ for some optimal cluster $C$, in terms of its standard deviation 
    by $\OPT(C)=|C|\sigma_C^2$, and the total cost of the optimal clustering is $\OPT=\sum_{C\in \COPT}|C|\sigma_C^2$.
\end{remark}

Given the aforementioned definitions, the main result that we show in this paper is the following theorem, whose formal proof appears in Section \ref{sec:thmproof}.
\begin{theorem}
\label{thm:main}
For any clustering $(C_1,C_2,\ldots ,C_k)$ of cost $\OPT$, and any $\alpha>2$, the $D^\alpha$ seeding procedure returns a clustering of expected cost at most 
\begin{equation*}
   O\left(f(\alpha)\cdot (g_\alpha)^{2/\alpha}\cdot \left( \frac{\smax}{\smin}\right)^{2-4/\alpha} \cdot \left(\min\{\ell,\log k\}\right)^{2/\alpha}\right)\cdot \OPT\ ,
\end{equation*} 
where $f(\alpha):=\frac{\alpha^2}{(\alpha/2-1)^{2/\alpha+1}}$. In particular, $f(\alpha)=O(\alpha^2)$ as $\alpha\rightarrow +\infty$. 
\end{theorem}

An immediate consequence of Theorem \ref{thm:main} is that $D^\alpha$ seeding with $\alpha > 2$ yields a constant-factor approximation guarantee for instances consisting of $k$ simplices of the same radius; a case that includes the described $\Omega(\log k)$ lower bound instances for the standard $D^2$ seeding from~\cite{arthur2007k}. Indeed, one can check that $g_\alpha = 2^\alpha$ in that case. We remark that the above stated guarantees do not require the clusters to be separated (as they are e.g. in the described lower bound instances of $D^2$ seeding).

We complement our results with the following lower bounds, which are fairly intuitive to prove.
 \begin{theorem}
\label{thm:lower_bounds}
There exists an instance with a clustering of cost $\OPT$ such that the $D^\alpha$ seeding procedure returns a clustering of expected cost at least 
\begin{equation*}
   \Omega \left(g_\alpha\right)^{2/\alpha}\cdot \OPT\ ,
\end{equation*} 
and another instance instance with a clustering of cost $\OPT$ such that the $D^\alpha$ seeding procedure returns a clustering of expected cost at least 
\begin{equation*}
   \Omega \left(\frac{\smax}{\smin}\right)^{2-4\alpha}\cdot \OPT\ .
\end{equation*} 
\end{theorem}
The formal proof of Theorem \ref{thm:lower_bounds} can be found in Appendix \ref{app:rmaxoverrmin} and Appendix \ref{app:galpha}. Finally, as mentioned in introduction, we also prove a lower bound on the greedy variant of \texttt{k-means++}.

\begin{theorem}
\label{thm:greedylowerbound}
There exists an instance with $k$ clusters for which $D^\alpha$ seeding guarantees a constant factor approximation in expectation for any fixed $\alpha>2$, and such that the greedy \texttt{$k$-means++} algorithm with $f(k)$ samples is not better than an $\Omega(\log \log f(k))^2$ approximation in expectation.
\end{theorem}
This highlights that $D^\alpha$ seeding can be superior in theory to the greedy variant. The proof of this last theorem can be found in Appendix \ref{sec:scikit}.

\subsection{Discussion on the parameters}
\label{sec:param}

 As we see, the guarantee of $D^{\alpha}$ seeding as stated in Theorem \ref{thm:main} has a dependence on $g_{\alpha}$, $\smax/\smin$, and $\min\{\ell,\log k\}$. Here we discuss a little more these dependencies.
 \paragraph{The parameter $g_\alpha$.} The ``moment'' condition can be seen as a characterization of the concentration of a cluster of points. For instance, one way that $g_{\alpha}$ could be non-constant is when the cluster has outliers that are still part of the cluster. On the contrary, if our clusters are generated by a Gaussian mixture, then $g_{\alpha}$ (for any constant $\alpha \ge 2$) is a constant (in fact, it is not difficult to compute and see that $g_\alpha\le \alpha^\alpha$ for Gaussian distributions). If we are in the infinite number of samples limit, where each cluster becomes defined by a density function $f$ on some domain $\mathcal D$, then $g_\alpha$ is equal to
 \begin{equation*}\label{eq:galphaint}
     \frac{\int_{\mathcal D} \int_{\mathcal D} ||x-y||^\alpha f(x)f(y)dxdy}{\left(\int_{\mathcal D} ||x-\mu||^2 f(x)dx\right)^{\alpha/2} }\ .
 \end{equation*}
Using a simple triangle inequality, we can see that 
\begin{align*}
     g_\alpha &= \frac{\int_{\mathcal D} \int_{\mathcal D} ||x-y||^\alpha f(x)f(y)dxdy}{\left(\int_{\mathcal D} ||x-\mu||^2 f(x)dx\right)^{\alpha/2} }\\
     &\le \frac{\int_{\mathcal D} \int_{\mathcal D} (||x-\mu||+||\mu-y||)^\alpha f(x)f(y)dxdy}{\left(\int_{\mathcal D} ||x-\mu||^2 f(x)dx\right)^{\alpha/2} }\\
     &\le 2^{\alpha+1}\cdot \frac{\int_{\mathcal D} ||x-\mu||^\alpha f(x)dx}{\left(\int_{\mathcal D} ||x-\mu||^2 f(x)dx\right)^{\alpha/2} }\ ,
 \end{align*}
    which is the $\alpha$-th standardized moment of the distribution \cite{wikipedia} (times $2^{\alpha+1}$). Of course, we loose a factor of $2^{\alpha+1}$ in this simple upperbound, but it will not matter much since only $(g_\alpha)^{2/\alpha}$ shows up in our guarantee. Let us denote by $\hat g_\alpha$ the above rough upperbound on $g_\alpha$. To obtain a better understanding of our guarantees, we give below the value of $(\hat g_\alpha)^{2/\alpha}$ for a few common distributions. W.l.o.g. we re-normalize to assume unit variance of each distribution (i.e. the denominator is equal to $1$ in the definition of $g_\alpha$).
    \begin{enumerate}
        \item Perhaps the most classic distribution is the Gaussian distribution with unit variance. In this case, the standardized moment is equal to $O((\alpha/2)^{\alpha/2})$ thus $(g_{\alpha})^{2/\alpha}$ is $O(\alpha^2)$. Furthermore, for multivariate Gaussians, $g_\alpha=O(\alpha^\alpha d^{\alpha}/2^{d/2})$. Meaning, in higher dimensions $g_\alpha$ decreases rapidly and this is well-supported by our understanding that a Gaussian distribution tends to become tightly concentrated around its mean in higher dimensions \cite{vershynin2020high}. In Remark \ref{rem:Gauss}, we detail how this $g_\alpha$ can be used to obtain the guarantees for mixture of Gaussians (as claimed in the abstract).
        \item For the exponential distribution, which has slightly fatter tails than Gaussian, $\texttt{Exp}(\lambda)$, the $\alpha^{th}$ moment is $O(\alpha!)$ and thus $(g_{\alpha})^{2/\alpha}$ is again $O(\alpha^2)$. 
        \item Now consider for instance, a univariate student-t distribution with degree of freedom $\nu > 0$, has its density function given by, $\frac{\Gamma((\nu+1)/2)}{\sqrt{\pi\nu}\Gamma(\nu/2)}\left(1+\frac{x^2}{\nu}\right)^{-((\nu+1)/2)}$. It is well known that $\alpha^{th}$ moment exists only if $\alpha < \nu$. Thus lower the degree of freedom, the heavier the tail gets, and $g_{\alpha}$ is bounded only when $\alpha < \nu$, in which case it is roughly $O(\nu^{\alpha})$ and thus $(g_{\alpha})^{2/\alpha} = O(\nu^2)$.
    \end{enumerate}

 We show in the Appendix \ref{app:galpha} that the dependency on $g_\alpha$ in Theorem \ref{thm:main} is tight. As a simple example that highlights the intuition, consider the instance given in Figure \ref{fig:G_alpha_example}. The red cluster is drawn from a standard $2$-dimensional Gaussian law. The blue cluster consists of many points highly concentrated at distance $\delta$ from the mean of the red cluster, and one single point at distance $\delta+\Delta$ from the mean of the red cluster. For this blue cluster, $g_\alpha$ will be unbounded when $\Delta$ tends to $\infty$ for any $\alpha>2$. Note that we can choose the parameters in this instance so that (i) both clusters have the same variance and (ii) both clusters have the same number of points so the other parameters do not play a role here. In this situation, there is still $1/2$ probability that the first center is selected in the red cluster (the first center is always chosen uniformly at random), and conditioned on that fact, the $D^\alpha$ seeding (for $\alpha>2$) will give way too much probability to the isolated point in the blue cluster, which is a serious issue. Our lower bound construction in the Appendix \ref{app:galpha} is a simple formalization of this intuition.

 \begin{figure}[h]
     \centering
     \includegraphics{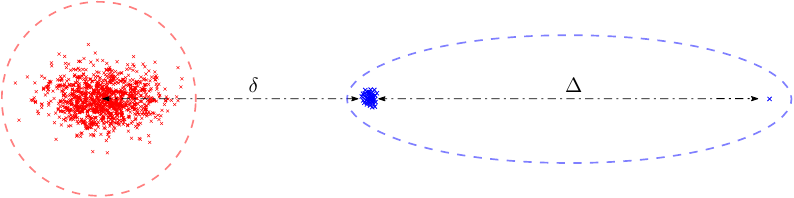}
     \caption{An instance with $k=2$.}
     \label{fig:G_alpha_example}
 \end{figure}
 
 \paragraph{The parameter $\smax/\smin$.} The dependence on $\smax/\smin$ is necessary and in fact tight (see Appendix \ref{app:rmaxoverrmin}). The main issue in sampling with large $\alpha$ is that the algorithm might sample repeatedly from a cluster with large standard deviation more often, and thus it might fail to discover some other clusters.
\paragraph{The parameter $\ell$.} The dependence on $\min\{\ell,\log k\}$ remains an intriguing open problem. Although it seems quite appealing that plugging in $\alpha=2$ we obtain $\min\{\ell,\log k\}^{2/\alpha}=O(\log k)$ as in the worst-case bound for \texttt{$k$-means++}, it is unclear to us if any dependency on $\ell$ or $k$ is needed when $\alpha>2$. Note that for $\alpha$ going to infinity, the parameter $\ell$ should matter less and less since $D^\alpha$ seeding becomes equivalent to picking the furthest point. This behavior is accurately reflected in our bound. 

\begin{remark}\label{rem:Gauss}
    The claim of an $O(\log \log k)^3$ approximation for mixture of Gaussians (as mentioned in the abstract) is now straightforward to see. More formally, suppose the mixture of $k$ Gaussians $\mathcal{X} \sim \sum_{i=1}^{k}w_i\mathcal{N}(\mu_i,\Sigma_i)$, satisfies $\max_{i,j \in [k]}\mbox{tr}(\Sigma_i)/\mbox{tr}(\Sigma_j)=O(1)$, with arbitrary mixing weights $w_i > 0$, for all $i \in [k]$ (w.l.o.g). Now, we can consider our reference clustering to be the points from the mixture (and this is optimal in the infinite sample limit). Then it is easy to see that $\smax/\smin = O(1)$, from the aforementioned trace condition. Furthermore, from the calculation of $g_\alpha$ for Gaussians (refer to \eqref{eq:galphaint}), our approximation guarantee becomes $O(\alpha^2 g_\alpha^{2/\alpha}\log^{2/\alpha} k))=O(\alpha^3\log^{2/\alpha} k)$. Thus, by setting $\alpha=\Theta(\log \log k)$, we get the desired approximation. 
\end{remark}

\subsection{On choosing \texorpdfstring{$\alpha$}{}}
\label{sec:choosing_alpha}
Our theorem states that there is a trade-off in choosing $\alpha$. We already know that $\alpha=2$ may not be the best choice and this is due to the well-clusterable instance of simplices of equal sidelength that are sufficiently far apart. However, Theorem \ref{thm:main} implies that any $\alpha > 2$ is a constant factor approximation in this case, and this is because $D^{\alpha}$ is more aggressive in discovering new clusters. But is it in our best interest to set $\alpha \rightarrow \infty$? Interestingly, Balcan et al. \cite{balcan2018data} show that the best $\alpha$ is learnable, hence selecting the best $\alpha$ is a task that is manageable when there is a training set. Moreover, our theorem predicts a new phenomenon that is not present in the experiments in \cite{balcan2018data}. For a mixture of balanced Gaussians, \cite{balcan2018data} obtain experimental results whose pattern roughly matches the one shown in left-hand side of Figure \ref{fig:exptGaussians}. This experiment corresponds to a Gaussian mixture with the \emph{same} covariance matrix (namely identity). However, our theory indicates that there is a dependence on the variances that is necessary and it appears in the right-hand side of Figure \ref{fig:exptGaussians}, where one of the Gaussians has much larger variance. Note that our approximation factor has dependence $(\smax/\smin)^{2-4/\alpha}$, and to mitigate this effect one can choose some $\alpha$ that is not too large but still greater than 2.

Furthermore, our result suggests an even simpler strategy. Using the training set, one can obtain an estimate of the key parameters $g_\alpha,\smax/\smin,\ell$, and use these estimates to select the best $\alpha$, see for example Figure \ref{fig:MNIST}. Our experiments seem to confirm this strategy. Note that the experiments mentioned here do not use additional steps of Lloyd's algorithm. As it is quite common to use this algorithm after the seeding, we run additional experiments in Section \ref{sec:expts}. Interestingly, we observe that the general pattern does not change much even after running Lloyd's algorithm until convergence.

Even when there is no training set, we might still have an idea of how the clusters should look like (for instance if the data is generated from a mixture of some distributions), and use this information to select the best $\alpha$. In practice, since the $D^{\alpha}$ seeding algorithm is simple and efficient to implement, even if there is only once instance, one can try a few different values of $\alpha$ and the select the best one.

\begin{figure}
    \centering
    \begin{minipage}{0.5\textwidth}
        \centering
        \includegraphics[width=\textwidth]{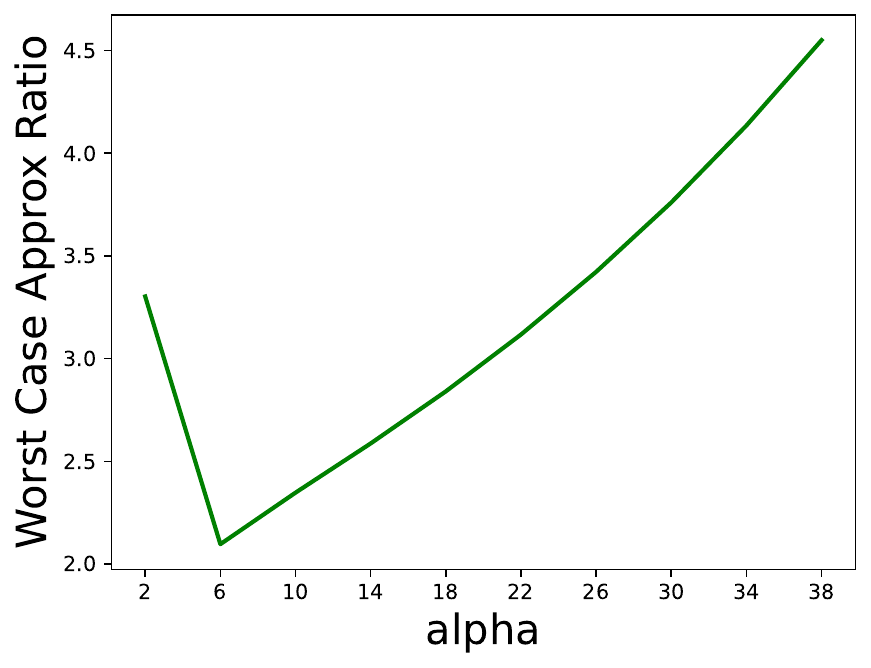} 
    \end{minipage}\hfill
    \begin{minipage}{0.5\textwidth}
        \centering
        \includegraphics[width=\textwidth]{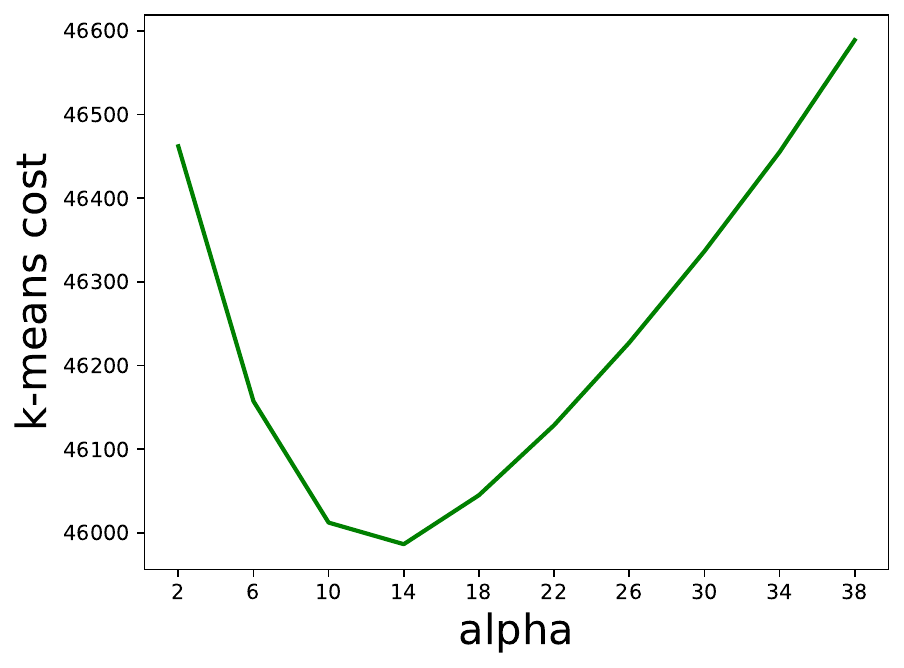} 
    \end{minipage}
    \caption{Theory vs Practice comparison on MNIST data set. The left-hand side figure shows approximation ratio predicted by our theoretical bounds, and the right-hand side figure shows the actual performance of $D^{\alpha}$ seeding.}
    \label{fig:MNIST}
\end{figure}

As a final note, we mention that it is a common wisdom among practitioners that $k$-means is a good objective, except when the clusters might have varying sizes and density, or when there are many outliers \cite{MLCourse}. In this context, ``varying sizes and density'' can clearly be interpreted as the parameters $\ell$ and $\smax/\smin$, while outliers clearly correspond to the parameter $g_\alpha$. If one believes this common wisdom, then our result essentially implies that whenever $k$-means is a good clustering objective, then choosing $\alpha>2$ should be almost always better than $\alpha=2$.

\begin{figure}
    \centering
    \begin{minipage}{0.5\textwidth}
        \centering
        \includegraphics[width=\textwidth]{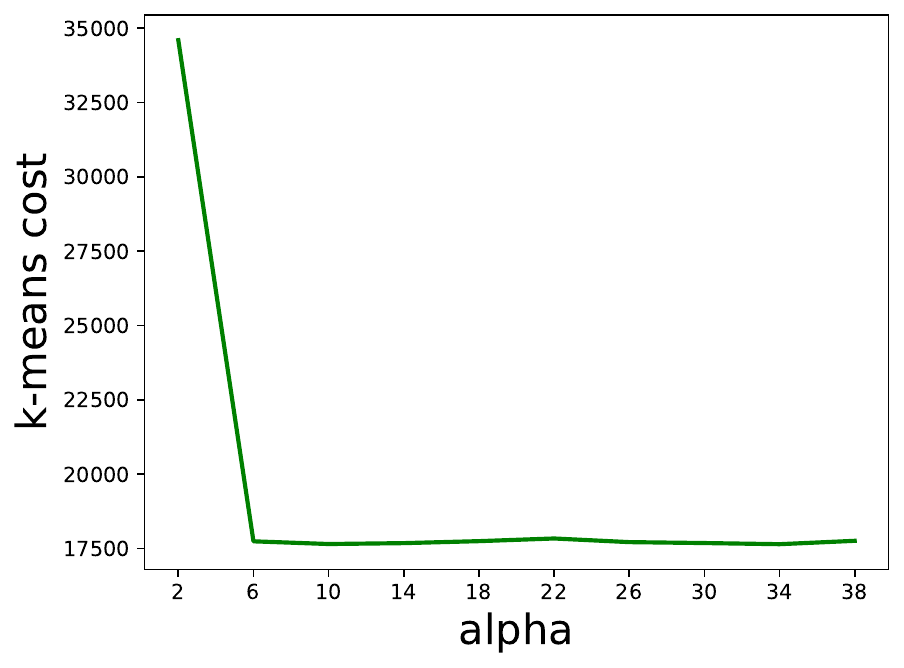} 
    \end{minipage}\hfill
    \begin{minipage}{0.5\textwidth}
        \centering
        \includegraphics[width=\textwidth]{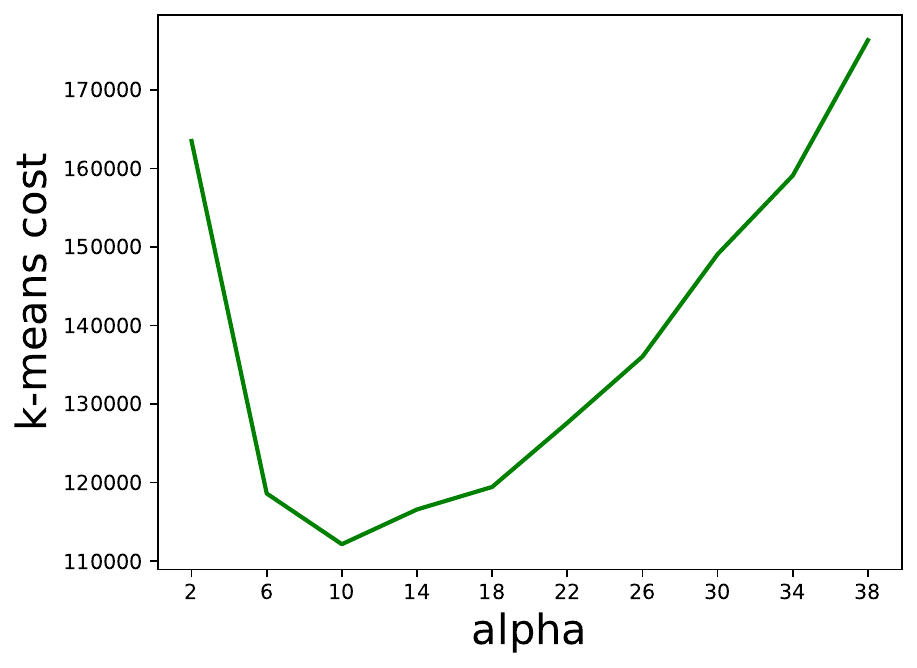} 
    \end{minipage}
    \caption{Performance of $D^{\alpha}$ seeding for two instances $\mathcal{D}_1$ (on the left) and $\mathcal{D}_2$ (on the right), from a balanced mixture of $k$ Gaussians with $k=4$ and $d=2$. The centers/means of the Gaussians for both instances are placed on the vertices of a square of side length $100$. However, the covariance matrices for $\mathcal{D}_1$ is $\{I,I,I,I\}$, whilst the covariance matrices for $\mathcal{D}_2$ is  $\{400I,I,I,I\}$, where $I$ is the identity matrix in two dimensions. }
    \label{fig:exptGaussians}
\end{figure}

\newpage
\section{Proof of Theorem \ref{thm:main}}\label{sec:thmproof}

The proof of Theorem \ref{thm:main} is inspired by a very clean potential function analysis of the $D^2$ seeding algorithm in \cite{sanjoyLN}. In a similar fashion, it is useful to bound the potential increase at each step. However, as we will see later, the potential function that is used for the $D^2$ seeding analysis does not seem to work for $D^{\alpha}$ seeding, and some additional ideas are required. We defer a more detailed discussion of the novelty of our potential function to Appendix \ref{sec:discussion_potential}.

As in \cite{sanjoyLN}, at every iteration $t$, it is useful to keep track of the set of optimal clusters from which a center has already been chosen (i.e. the \textit{hit} clusters) and the complement of this set which is the set of \textit{undiscovered} clusters. Formally, we define $H_t$ to be the set of \textit{hit} clusters after selecting $t$ centers denoted by $Z_t$, i.e. 
\begin{equation*}
    H_t:=\{C\in \COPT: C \cap Z_t \neq \emptyset\}\ .
\end{equation*}
$U_t$ is defined to be the set of remaining \textit{undiscovered} clusters, i.e, $U_t:=\COPT \setminus H_t$. Furthermore, we define $\cost_t (C)$ as a shorthand to denote the cost induced by the points in the cluster $C$, after the set $Z_t$ of $t$ centers are chosen. More formally,
\begin{equation*}
    \cost_t(C) := \cost(C,Z_t)= \sum_{x\in C}\min_{z\in Z_t}\lVert x-z \rVert^2\ .
\end{equation*}
Since we analyze the $D^\alpha$ seeding, we also need to work with the $\alpha$-cost:
\begin{equation*}
    \costa_t(C) :=\costa_t(C,Z_t)= \sum_{x\in C}\min_{z\in Z_t}\lVert x-z \rVert^{\alpha}\ .
\end{equation*}
For any set $S$ of clusters, we define $\cost_t(S):=\sum_{C\in S}\cost_t(C)$, and $\costa_t (S)=\sum_{C\in S}\costa_t(C)$. Moreover we can talk about the cost of a single point at iteration $t$ as $\cost_t(x):=\min_{z \in Z_t}\lVert x-z\rVert^2$.

\paragraph{The potential function.} Now we proceed to define our potential function which will be used to upperbound the cost of undiscovered clusters. For each $i\ge 0$, we define $S_i$ to be the set of clusters in $C\in \COPT$ such that $|C|$ lies in the interval $[2^i,2^{i+1})$ (recall that we also defined $k_i:=|S_i|$). For each $i\ge 0, t\ge 0$, we define an integer $\tau_i(t)\ge 0$ which will be a local counter, relevant only for the clusters in $S_i$ at iteration $t$. $z_t$ is defined to be the center selected at iteration $t$, and $U_t^{(i)}:=U_t\cap S_i$ the set of undiscovered clusters in $S_i$. For each $i\ge 0$ and time $t$, we will define an integer $w_i(t)\ge 0$ corresponding to the number of iterations that are considered \textit{wasted} by the clusters in $S_i$ at time $t$. We use the word \textit{wasted} to follow the intuition given in \cite{sanjoyLN} where an iteration $t$ is wasted when the selected center $z_t$ belongs to an already discovered center (in particular this new center does not discover a new cluster).

\paragraph{Some intuition.} In our potential function, the counter $\tau_i$ will intuitively count how many iterations were relevant to the set of clusters $S_i$. Once $\tau_i$ reaches the value $k_i$, we will consider that the set $S_i$ was given enough tries to cover its clusters. Of course, during the sampling, it might be that some cluster is hit twice, which will increase the wasted counter of \textit{all} the counters $(\tau_j)_{j\ge 0}$. In light of the proof in \cite{sanjoyLN}, it might be counter intuitive that when a cluster in $S_i$ is hit twice then the wasted counter $w_j$ increases even if $j\neq i$. However, it will become apparent later why we proceed like this.\\

Formally, we initialize $t=0$, $w_i(0)=0$ for all $i$, and $\tau_i(0)=0$ for all $i$. Then, we maintain these quantities as follows. At time $t\ge 0$, let $C$ be the cluster the next center $z_t$ is chosen from. Let $i\ge 0$ be the integer such that $C\in S_i$. Then, there are two cases. 

\begin{enumerate}
    \item If $C\in U_t^{(i)}$, then we set (for all $j\ge 0$)
    \begin{equation*}
    \tau_j(t+1)=
    \begin{cases}
        \tau_j(t)+1 \mbox{ if $j=i$ and $\tau_j(t)<k_j$}\\
        \tau_j(t) \mbox{ otherwise.}
    \end{cases}
\end{equation*}
and $w_j(t+1)=w_j(t) $ for all $j\ge 0$.
\item Otherwise if $C\in H_t$, then we set (for all $j\ge 0$)
\begin{equation*}
    \tau_j(t+1)=
    \begin{cases}
        \tau_j(t)+1 \mbox{ if $\tau_j(t)<k_j$}\\
        \tau_j(t) \mbox{ if $\tau_j(t)\ge k_j$ ,}\\
    \end{cases}
\end{equation*}
and 
\begin{equation*}
    w_j(t+1)=
    \begin{cases}
        w_j(t)+1 \mbox{ if $\tau_j(t)<k_j$}\\
        w_j(t) \mbox{ if $\tau_j(t)\ge k_j$ .}\\
    \end{cases}
\end{equation*}
\end{enumerate}
Based on these quantities, we define the potential function as follows. First, we define
\begin{equation}
    \phi_i(t) := \frac{w_i(t)}{|U_t^{(i)}|} \cdot (2^i)^{1-2/\alpha} \cdot \sum_{C\in U_t^{(i)}} (\costa_t (C))^{2/\alpha}\ .
\end{equation}
The final potential function can now be defined as 
\begin{equation}
    \phi(t) := \sum_{i\ge 0} \phi_i(t)\ .
\end{equation}

Using this potential, we are ready to proceed with the main proof. Akin to Dasgupta's analysis \cite{sanjoyLN}, we split the proof in three main parts. In Section \ref{sec:potential_and_final_cost}, we show that $\phi(k)$ is indeed an upper bound on the final cost of undiscovered clusters. In Section \ref{sec:hit_clusters}, we upper bound the cost of hit clusters using $g_\alpha$. In Section \ref{sec:increase_potential} we upper-bound the increase of the potential function. Finally, we complete the puzzle in Section \ref{sec:wrapping_it_up}.

\subsection{Relating the potential and the cost of undiscovered clusters}
\label{sec:potential_and_final_cost}
This part is fairly straightforward, using a few lemmas.
\begin{lemma}
\label{lem:number_of_tries}
    For all $i\ge 0$, we have that $w_i(k)\ge|U_k^{(i)}|$.
\end{lemma}
\begin{proof}
    Consider the quantity $\delta_i(t):=|U_t^{(i)}|-w_i(t)$. Clearly, we have that $\delta_i(0)=k_i$. Next, notice that for every time-step $t<k$ such that $\tau_i(t+1)=\tau_i(t)+1$, the quantity $\delta_i(t)$ decreases by $1$. Indeed, either the algorithm discovers a new cluster in $|U_t^{(i)}|$ (in which case $|U_{t+1}^{(i)}|=|U_t^{(i)}|-1$), or the algorithm wastes an iteration in some cluster, in which case $w_i(t+1)=w_i(t)+1$. 
    
    Finally, we claim that there are at least $k_i$ such iterations. Note that if $\tau_i(t)<k_i$, then the only way we have that $\tau_i$ does not increase is if the algorithm discovers a new cluster in $|U_t^{(j)}|$ for some $j\neq i$. This can happen at most $\sum_{j\neq i} k_j$ times. Hence there must be at least $k-\sum_{j\neq i} k_j=k_i$ iterations where the counter $\tau_i$ increases. If we denote by $t_i$ the first time at which $\tau_i(t_i)=k_i$, this implies $t_i\le k$ and that that $\delta_i(t_i)=0$ hence $w_i(t_i)=|U_{t_i}^{(i)}|$. Finally, note that for $t>t_i$, $w_i(t)$ does not change anymore, and $|U_t^{(i)}|$ can only decrease; which concludes the proof.
\end{proof}

\begin{lemma}
\label{lem:potential_upper_bound}
    We have that $\phi(k) \ge \sum_{i\ge 0}\cost (U_k^{(i)},Z_k)/2=\cost_k (U_k)/2$.
\end{lemma}
\begin{proof}
    Using Lemma \ref{lem:number_of_tries}, we have that 
    \begin{equation*}
        \phi(k)=\sum_{i\ge 0}\phi_i(k)\ge (1/2)\cdot (2^{i+1})^{1-2/\alpha} \cdot \sum_{C\in U_k^{(i)}} (\costa (C,Z_k))^{2/\alpha}\ .
    \end{equation*}
    Using Jensen's inequality (see Appendix \ref{sec:inequalities} for a reference) and the fact that $|C|\le 2^{i+1}$ for all $C\in U_k^{(i)}$, we have that 
    \begin{equation*}
        (2^{i+1})^{1-2/\alpha}(\costa (C,Z_k))^{2/\alpha} \ge \cost_k(C)
    \end{equation*}
    for all $C\in U_k^{(i)}$, which concludes the proof.
\end{proof}

\subsection{The cost of hit clusters}
\label{sec:hit_clusters}
In this section, we give upper bounds on the expected cost of clusters that were hit during the seeding process. These proofs are similar to the ones found in \cite{sanjoyLN, arthur2007k}. In fact, the first two lemmas are taken directly from \cite{arthur2007k}.
\begin{lemma}[From \cite{arthur2007k}]
\label{lem:arthur_cost_hit_alpha_sampling}
Assume some arbitrary set $T$ of centers have already been selected, and $z\in C$ is added next using $D^\alpha$-sampling. Then,
    \begin{equation*}
        \mathbb  E\left[ \costa(C,T\cup \{z\})\mid z\in C\right]\le 2^{2\alpha}\cdot   \costa(C,\mu_C) \ .
    \end{equation*}
\end{lemma}
\begin{lemma}[From \cite{arthur2007k}]
\label{lem:arthur_cost_hit_uniform_sampling}
Assume some point $z$ is selected uniformly at random among the points belonging to some cluster $C$. Then, for any $\alpha\ge 2$,
    \begin{equation*}
        \mathbb  E\left[ \costa(C,z)\right]\le 2^{\alpha}\cdot   \costa(C,\mu_C) \ .
    \end{equation*}
\end{lemma}

The next lemma deals with the expected squared cost of the hit clusters during the seeding process.
\begin{lemma}
\label{lem:cost_hit_clusters_alpha_sampling}
    Assume some arbitrary set $T$ of centers have already been selected, and $z\in C$ is added next using $D^\alpha$ seeding. Then,
    \begin{equation*}
        \mathbb  E\left[ \cost(C,T\cup \{z\})\mid z\in C\right]\le (4e+(\alpha+1)^2\cdot (g_\alpha)^{2/\alpha} )\cdot  \cost(C,\mu_C) \ .
    \end{equation*}
\end{lemma}
\begin{proof}
    We can write
    \begin{equation*}
        \mathbb E\left[ \cost(C,T\cup \{z\})\mid z\in C\right] = \sum_{z\in C} \frac{\costa (z,T)}{\costa (C,T)}\cdot \sum_{x\in C} \min \{\cost (x,T), ||x-z||^2\}\ .
    \end{equation*}
    Let us upperbound the quantity $\costa (z,T)$. For this, let us fix any $x\in C$. Then, if we denote by $t_x$ the point in $T$ which is closest to $x\in C$, we have that 
    \begin{equation*}
        \costa (z,T) \le  ||z-t_x||^\alpha \le (||z-x||+||x-t_x||)^\alpha \le (\alpha+1)^\alpha \cdot ||z-x||^\alpha + (1+1/\alpha)^\alpha\cdot  ||x-t_x||^\alpha\ ,
    \end{equation*}
    using the triangle inequality and a case distinction whether $||z-x||>(1/\alpha) ||x-t_x||$ or not. Averaging this upper bound over all $x\in C$, we obtain that 
    \begin{align*}
        \costa (z,T) &\le \frac{(\alpha+1)^\alpha \cdot \costa(C,z)}{|C|}+\frac{(1+1/\alpha)^\alpha\cdot \costa(C,T)}{|C|}\\
        &\le \frac{(\alpha+1)^\alpha\cdot  \costa(C,z)}{|C|}+\frac{e\cdot \costa(C,T)}{|C|}\ .
    \end{align*}
    Hence we can rewrite 
    \begin{align*}
        \mathbb E&\left[ \cost(C,T\cup \{z\})\mid z\in C\right] \\
        &\le \sum_{z\in C} \frac{\frac{(\alpha+1)^\alpha\cdot  \costa(C,z)}{|C|}+\frac{e\cdot \costa(C,T)}{|C|}}{\costa (C,T)}\cdot \sum_{x\in C} \min \{\cost (x,T), ||x-z||^2\}\\
        &\le \frac{e}{|C|}\sum_{z\in C}  \cost(C,z) + \frac{(\alpha+1)^\alpha}{|C|}\sum_{z\in C} \frac{\costa(C,z)}{\costa (C,T)}\cdot \cost (C,T)\ .
    \end{align*}
    To finish the argument, note that the first term in the last line corresponds to the expected cost of $C$ if we pick once center $z\in C$, uniformly at random. By Lemma \ref{lem:arthur_cost_hit_uniform_sampling}, this is at most $(4e)\cdot \cost(C,\mu_C)$. For the second term, we use Equation \eqref{eq:concentration_assumption_new} to write 
    \begin{equation*}
        \frac{(\alpha+1)^\alpha}{|C|}\sum_{z\in C} \frac{\costa(C,z)}{\costa (C,T)}\cdot \cost (C,T) \le ((\alpha+1)^\alpha g_\alpha) \cdot  (\cost(C,\mu_C))^{\alpha/2}\cdot  \frac{\cost(C,T)}{|C|^{\alpha/2-1}\costa(C,T)}\ .
    \end{equation*}
    Using Jensen's inequality, we obtain that $|C|^{\alpha/2-1}\costa(C,T) \ge (\cost(C,T))^{\alpha/2}$. Hence we finally get that 

    \begin{equation*}
        \frac{(\alpha+1)^\alpha}{|C|}\sum_{z\in C} \frac{\costa(C,z)}{\costa (C,T)}\cdot \cost (C,T) \le ((\alpha+1)^\alpha g_\alpha) \cdot  \frac{(\cost(C,\mu_C))^{\alpha/2}}{(\cost(C,T))^{\alpha/2-1}}\ .
    \end{equation*}
    From there, if $\cost(C,T)\le (\alpha+1)^2\cdot (g_\alpha)^{2/\alpha}\cost(C,\mu_C)$ then the lemma already holds since adding an additional center can only decrease the cost of $C$. Otherwise, we clearly get that 
    \begin{equation*}
        \frac{(\alpha+1)^\alpha}{|C|}\sum_{z\in C} \frac{\costa(C,z)}{\costa (C,T)}\cdot \cost (C,T) \le ((\alpha+1)^2\cdot (g_\alpha)^{2/\alpha}) \cdot  \cost(C,\mu_C)\ ,
    \end{equation*}
    which finishes the proof.
\end{proof}

The last lemma relates the squared cost and the $\alpha$-powered cost of any cluster.
\begin{lemma}
\label{lem:cost_alpha_clusters}
    For any cluster $C$, we have that 
    \begin{equation*}
        \costa(C,\mu_C)\le g_\alpha\cdot |C|\cdot (\sigma_C)^\alpha\ . 
    \end{equation*}
\end{lemma}
\begin{proof}
    We note that $\costa(C,\mu_C)\le \frac{1}{|C|}\sum_{z\in C} \costa(C,z)$ (using Jensen's inequality, and the convexity of the function $y\mapsto \sum_{z\in C}||z-y||^\alpha$). Hence, we obtain
    \begin{align*}
        \costa(C,\mu_C) &\le \left(\frac{1}{|C|}\sum_{z\in C} \costa(C,z)\right)\\
        &\le |C|^{1-\alpha/2} \cdot g_\alpha \cdot (\cost(C,\mu_C))^{\alpha/2}= g_\alpha\cdot |C|\cdot (\sigma_C)^\alpha\ ,
    \end{align*}
    where the second inequality uses our assumption.
\end{proof}

\subsection{The increase of potential}
\label{sec:increase_potential}
In this section, we bound the final potential $\phi(k)$. First, we analyze the increase of local potential $\phi_i$ individually, then we use these results to bound the final expected potential $\mathbb E [\phi(k)]$.

\subsubsection{The increase in a weight class}

\begin{lemma}
\label{lem:potential_hit_old}
    Let $B_t$ be the event that the $t$-th center is selected from $H_t$. Then, for any $t> 0$, $i\ge 0$, and any past choice of centers $Z_{t-1}$, we have that 
    \begin{equation*}
    \mathbb E\bigg[\phi_i(t)-\phi_i(t-1) \mid \{Z_{t-1}, B_t\}\bigg] \le (\tau_i(t)-\tau_i(t-1))\cdot \frac{(2^i)^{1-2/\alpha}}{|U_{t-1}^{(i)}|} \cdot \sum_{C\in U_{t-1}^{(i)}} (\costa_{t-1} (C))^{2/\alpha}\ .
\end{equation*}
\end{lemma}
\begin{proof}
    Fix some $i\ge 0$. If $\tau_i(t-1)=\tau_i(t)$, we have by definition $w_i(t)=w_i(t-1)$ and clearly the potential $\phi_i$ cannot increase. Otherwise, we simply note that $w_i(t)=w_i(t-1)+1$, and the result clearly follows. Note that in both cases, we use the fact that $(\costa_t (C))^{2/\alpha}\le (\costa_{t-1} (C))^{2/\alpha}$.
\end{proof}

\begin{lemma}
\label{lem:potential_hit_new}
Let $A_t^{(i)}$ be the event that the $t$-th center is selected from and undiscovered cluster belonging to the weight class $i$. Then, for any $t> 0$, $i,j\ge 0$, any past choice of centers $Z_{t-1}$, we have that 
\begin{equation*}
    \mathbb E\bigg[\phi_j(t)-\phi_j(t-1) \mid \{Z_{t-1}, A_t^{(i)}\}\bigg] \le 0.
\end{equation*}
\end{lemma}
\begin{proof}
    We first note that since $z_t\notin H_t$, it is clear that $\tau_j(t)=\tau_j(t-1)$ for all $j\neq i$, hence $\phi_j(t)\le \phi_j(t-1)$ for all $j\neq i$. To bound the change of $\phi_i$ we notice that $w_i(t)=w_i(t-1)$, and that $|U_t^{(i)}|=|U_{t-1}^{(i)}|-1$. Let $C_z$ be the cluster which is selected in $U_{t-1}^{(i)}$. We claim that 
    \begin{equation}
        \mathbb E\bigg[(\costa_{t-1}(C_z))^{2/\alpha} \mid  \{Z_{t-1}, A_t^{(i)}\}\bigg] \ge \frac{1}{|U_{t-1}^{(i)}|}\cdot \sum_{C\in U_{t-1}^{(i)}}(\costa_{t-1}(C))^{2/\alpha}\ .
    \end{equation}
    Indeed, note that 
    \begin{align*}
       \E\left[(\costa_{t-1}(C_z))^{2/\alpha}\mid \{Z_{t-1}, A_t^{(i)}\} \right]&=\sum_{C \in U_{t-1}^{(i)}}\frac{\costa_{t-1}(C)}{\sum_{C \in U_{t-1}^{(i)}}\costa_{t-1}(C)}\cdot\left(\costa_{t-1}(C)\right)^{2/\alpha}\\
       &\hspace{-20mm}=\frac{|U_{t-1}^{(i)}|}{\sum_{C \in U_{t-1}^{(i)}}\costa_{t-1}(C)}\cdot \sum_{C \in U_{t-1}^{(i)}}\frac{\costa_{t-1}(C)\cdot\left(\costa_{t-1}(C)\right)^{2/\alpha}}{|U_{t-1}^{(i)}|}\\
       &\hspace{-20mm}\stackrel{(1)}{\geq} \frac{|U_{t-1}^{(i)}|}{\sum_{C \in U_{t-1}^{(i)}}\costa_{t-1}(C)}\cdot \sum_{C \in U_{t-1}^{(i)}}\frac{\costa_{t-1} (C)}{|U_{t-1}^{(i)}|}\cdot \sum_{C \in U_{t-1}^{(i)}}\frac{(\costa_{t-1}(C))^{2/\alpha}}{|U_{t-1}^{(i)}|}\\
       &\hspace{-20mm}=\sum_{C \in U_{t-1}^{(i)}}\frac{(\costa_{t-1}(C))^{2/\alpha}}{|U_{t-1}^{(i)}|}\ .
   \end{align*}
   The inequality (1) can be obtained using Lemma \ref{lem:chebyshevsum} (Chebyshev's sum inequality) by considering the ordered sequence $\left(\costa_{t-1}(C)\right)_{C\in U_{t-1}^{(i)}}$ and $\left((\costa_{t-1}(C))^{2/\alpha}\right)_{C \in U_{t-1}^{(i)}}$. Thus we have:
   \begin{align*}
   &\E\left[\phi_i(t)\mid \{Z_{t-1}, A_{t}^{(i)}\}\right]\\
   &\leq \frac{w_i(t)}{|U_{t-1}^{(i)}|-1} \cdot (2^i)^{1-2/\alpha} \cdot \left( \sum_{C\in U_{t-1}^{(i)}} (\costa_t (C))^{2/\alpha} -  \E\left[(\costa_{t-1}(C_z))^{2/\alpha}\mid \{Z_{t-1}, A_t^{(i)}\} \right] \right)\\
   &\le \frac{w_i(t)}{|U_{t-1}^{(i)}|-1} \cdot (2^i)^{1-2/\alpha} \cdot \left( \frac{|U_{t-1}^{(i)}|-1}{|U_{t-1}^{(i)}|}\cdot \sum_{C\in U_{t-1}^{(i)}} (\costa_{t-1} (C))^{2/\alpha}\right)\\
   &=\phi_i(t-1).
   \end{align*} 
\end{proof}

\begin{lemma}
\label{lem:potential_increase}
For every $i\ge 0, \alpha>2, t>0$, we have that 
\begin{equation*}
    \mathbb E\left[\phi_i(t)-\phi_i(t-1)  \mid \{\tau_i(t)=\tau_i(t-1)+1\}\right] \le h(\alpha)\cdot  \left(2^i\right)^{1-2/\alpha} \cdot (k_i-\tau_i(t-1))^{2/\alpha} \cdot \left((2^{2\alpha} g_\alpha) \sum_{C\in \mathcal{C}_{\OPT}}  |C|(\sigma_C)^\alpha\right)^{2/\alpha}\ ,
\end{equation*}
where $h(\alpha)=\frac{(\alpha/2-1)^{1-2/\alpha}}{\alpha/2}$, and 
\begin{equation*}
    \mathbb E\left[\phi_i(t)-\phi_i(t-1)  \mid \{\tau_i(t)=\tau_i(t-1)\}\right] \le 0\ .
\end{equation*}
\end{lemma}
\begin{proof}
    If $\tau_i(t)=\tau_i(t-1)$, then either we can apply Lemma \ref{lem:potential_hit_old} or Lemma \ref{lem:potential_hit_new}, and in both cases the expected potential $\phi_i$ can only decrease. 
    
    If $\tau_i(t)=\tau_i(t-1)+1$ then either $z_t$ belongs to $U_{t-1}^{(i)}$, in which case the potential $\phi_i$ can only decrease in expectation (by Lemma \ref{lem:potential_hit_new} again). Therefore there remains only the case that $z_t\in H_{t-1}$ (we denote this event by $B_t$). In this case, let us denote by $Z_{t-1}$ the current set of selected centers. Using Lemma \ref{lem:potential_hit_old} and the previous cases we have that

    \begin{align*}
        \mathbb E&\left[\phi_i(t)-\phi_i(t-1) \mid \{\tau_i(t)=\tau_i(t-1)+1,Z_{t-1}\}\right]\\
        &\le (2^i)^{1-2/\alpha} \cdot \frac{\sum_{C\in U_{t-1}^{(i)}} (\costa_{t-1} (C))^{2/\alpha}}{|U_{t-1}^{(i)}|}\cdot \mathbb P[B_t \mid \{\tau_i(t)=\tau_i(t-1)+1,Z_{t-1}\}]\\
        &\le (2^i)^{1-2/\alpha} \cdot \frac{\sum_{C\in U_{t-1}^{(i)}} (\costa_{t-1} (C))^{2/\alpha}}{|U_{t-1}^{(i)}|}\cdot \frac{\costa_{t-1}(H_{t-1})}{\costa_{t-1}(H_{t-1})+ \costa_{t-1}(U_{t-1}^{(i)})}\\
        &\le (2^i)^{1-2/\alpha} \cdot \frac{\sum_{C\in U_{t-1}^{(i)}} (\costa_{t-1} (C))^{2/\alpha}}{|U_{t-1}^{(i)}|}\cdot \frac{\costa_{t-1}(H_{t-1})}{\costa_{t-1}(H_{t-1})+ |U_{t-1}^{(i)}| \cdot \left(\frac{\sum_{C\in U_{t-1}^{(i)}} (\costa_{t-1} (C))^{2/\alpha}}{|U_{t-1}^{(i)}|}\right)^{\alpha/2}}\ ,
    \end{align*}
    where the last inequality uses Jensen's inequality. If we consider the last expression as a function of $X:=\frac{\sum_{C\in U_{t-1}^{(i)}} (\costa_{t-1} (C))^{2/\alpha}}{|U_{t-1}^{(i)}|}$ (the other quantities being fixed), one can see that this expression attains a maximum value for
    \begin{equation*}
        X=\left(\frac{\costa_{t-1}(H_{t-1})}{|U_{t-1}^{(i)}|\cdot (\alpha/2-1)}\right)^{2/\alpha}\ .
    \end{equation*}
    Plugging in this value, we obtain that 
    \begin{align*}
        \mathbb E&\left[\phi_i(t)-\phi_i(t-1)  \mid \{\tau_i(t)=\tau_i(t-1)+1,Z_{t-1}\}\right] \\
        &\le \frac{(\alpha/2-1)^{1-2/\alpha}}{\alpha/2} \cdot (2^{i})^{1-2/\alpha} \cdot \frac{(\costa_{t-1}(H_{t-1}))^{2/\alpha}}{|U_{t-1}^{(i)}|^{2/\alpha}}\\
        &\le  \frac{(\alpha/2-1)^{1-2/\alpha}}{\alpha/2} \cdot (2^{i})^{1-2/\alpha} \cdot (k_i-\tau_i(t-1))^{-2/\alpha} \cdot \mathbb (\costa_{t-1}(H_t))^{2/\alpha}\ ,
        \end{align*}
    where the second inequality uses the fact that $|U_{t-1}^{(i)}|\ge k_i-\tau_i(t-1)$. By the law of total expectation, we have that
    \begin{align*}
        \mathbb E&\left[\phi_i(t)-\phi_i(t-1)  \mid \{\tau_i(t)=\tau_i(t-1)+1\} \right] = \mathbb E_{Z_{t-1}}\left[\mathbb E\left[\phi_i(t)-\phi_i(t-1)  \mid \{\tau_i(t)=\tau_i(t-1)+1,Z_{t-1}\}\right]\right]\\
        &\le \frac{(\alpha/2-1)^{1-2/\alpha}}{\alpha/2} \cdot (2^{i})^{1-2/\alpha} \cdot (k_i-\tau_i(t-1))^{-2/\alpha} \cdot \mathbb E_{Z_{t-1}}\left[\left( \costa_{t-1}(H_{t-1})\right)^{2/\alpha}\right]\\
        &\le \frac{(\alpha/2-1)^{1-2/\alpha}}{\alpha/2} \cdot (2^{i})^{1-2/\alpha} \cdot (k_i-\tau_i(t-1))^{-2/\alpha} \cdot  \left( \mathbb E_{Z_{t-1}}\left[\costa_{t-1}(H_{t-1})\right]\right)^{2/\alpha}\\
        &\le \frac{(\alpha/2-1)^{1-2/\alpha}}{\alpha/2} \cdot (2^{i})^{1-2/\alpha} \cdot (k_i-\tau_i(t-1))^{-2/\alpha} \cdot \left((2^{2\alpha} g_\alpha) \sum_{C\in \mathcal C_\OPT}  |C|(\sigma_C)^\alpha \right)^{2/\alpha}\ ,
    \end{align*}
    where the second inequality uses Jensen's inequality, and the last inequality uses Lemmas \ref{lem:cost_alpha_clusters}, \ref{lem:arthur_cost_hit_uniform_sampling}, and Lemma \ref{lem:arthur_cost_hit_alpha_sampling}. Indeed, we clearly have the expected $\alpha$-powered cost of a hit cluster is at most its expected cost after the first time it is hit, which, using Lemmas \ref{lem:arthur_cost_hit_alpha_sampling} and \ref{lem:arthur_cost_hit_uniform_sampling} is at most $2^{2\alpha}$ times $\costa(C,\mu_C)$.
\end{proof}
\subsubsection{The global increase}
In this part, we are ready to bound the final expected value of $\phi(k)$ with the following lemma.
\begin{lemma}
\label{lem:potential_global}
For any $\alpha>2$, we have that 
\begin{equation*}
    \frac{\mathbb E\left[\phi(k)\right]}{\OPT} \le f(\alpha) \cdot (g_\alpha)^{2/\alpha} \cdot \left(\frac{\smax}{\smin}\right)^{2-4/\alpha}\cdot  \min \{\ell,\log(k)\}^{2/\alpha}\ ,
\end{equation*}
where $f(\alpha)=\frac{16(\alpha/2-1)^{1-2/\alpha}}{\alpha/2-1}\cdot \frac{2-2^{2/\alpha-1}}{1-2^{2/\alpha-1}}$.
\end{lemma}
\begin{proof}
    Using Lemma \ref{lem:potential_increase}, we obtain that 
    \begin{align*}
        \mathbb E\left[ \phi_i(k) \right] &\le h(\alpha)\cdot  \left(2^i\right)^{1-2/\alpha} \cdot \left((2^{2\alpha} g_\alpha) \sum_{C\in \mathcal C_\OPT}  |C|(\sigma_C)^\alpha\right)^{2/\alpha} \cdot \sum_{t=0}^{k_i-1} (k_i-t)^{-2/\alpha}  \\
        &\le h(\alpha)\cdot  \left(2^i\right)^{1-2/\alpha} \cdot \left((2^{2\alpha} g_\alpha) \sum_{C\in \mathcal C_\OPT}  |C|(\sigma_C)^\alpha\right)^{2/\alpha} \cdot \int_{0}^{k_i} (k_i-u)^{-2/\alpha} du \\
        &\le h(\alpha)\cdot  \left(2^i\right)^{1-2/\alpha} \cdot \left((2^{2\alpha} g_\alpha) \sum_{C\in \mathcal C_\OPT}  |C|(\sigma_C)^\alpha\right)^{2/\alpha} \cdot \frac{(k_i)^{1-2/\alpha}}{1-2/\alpha}\ .
    \end{align*}
    Therefore, we obtain  
    \begin{align*}
        \mathbb E[\phi(k)] = \sum_{i\ge 0} \mathbb E[\phi_i(k)] &\le \frac{h(\alpha)}{1-2/\alpha}\cdot \left((2^{2\alpha} g_\alpha) \sum_{C\in \mathcal C_\OPT}  |C|(\sigma_C)^\alpha\right)^{2/\alpha} \cdot \sum_{i\ge 0} (2^i k_i)^{1-2/\alpha}\\
        &= \frac{8h(\alpha) (g_\alpha)^{2/\alpha}}{1-2/\alpha}\cdot  \left(\sum_{C\in \mathcal C_\OPT}  |C|(\sigma_C)^\alpha\right)^{2/\alpha} \cdot \sum_{i\ge 0} (2^i k_i)^{1-2/\alpha}\ .
    \end{align*}
    Comparing this with optimum clustering, we get using basic algebraic manipulations
    \begin{align*}
        \frac{\mathbb E[\phi(k)]}{\OPT} &= \frac{8h(\alpha) (g_\alpha)^{2/\alpha}}{1-2/\alpha} \cdot \frac{\left(\sum_{C}  |C|(\sigma_C)^\alpha\right)^{2/\alpha}}{\sum_{C}  |C|(\sigma_C)^2} \cdot  \sum_{i\ge 0} (2^i k_i)^{1-2/\alpha}\\
        &= \frac{8h(\alpha) (g_\alpha)^{2/\alpha}}{1-2/\alpha} \cdot \frac{\left(\sum_{C}  |C|\left(\frac{\sigma_C}{\smax}\right)^\alpha \smax^\alpha \right)^{2/\alpha}}{\sum_{C}  |C|\left(\frac{\sigma_C}{\smin}\right)^2\cdot \smin^2} \cdot  \sum_{i\ge 0} (2^i k_i)^{1-2/\alpha}\\
        &\le \frac{8h(\alpha) (g_\alpha)^{2/\alpha}}{1-2/\alpha} \cdot \frac{\left(\sum_{C}  |C|\left(\frac{\sigma_C}{\smax}\right)^2 \right)^{2/\alpha}\cdot \smax^2}{\sum_{C}  |C|\left(\frac{\sigma_C}{\smin}\right)^2\cdot \smin^2} \cdot  \sum_{i\ge 0} (2^i k_i)^{1-2/\alpha}\\
        &\le \frac{8h(\alpha) (g_\alpha)^{2/\alpha}}{1-2/\alpha} \cdot \frac{\left(\sum_{C}  |C|\left(\sigma_C\right)^2 \right)^{2/\alpha}\cdot \smax^{2-4/\alpha}}{\sum_{C}  |C|\left(\frac{\sigma_C}{\smin}\right)^2\cdot \smin^2} \cdot  \sum_{i\ge 0} (2^i k_i)^{1-2/\alpha}\\
        &\le \frac{8h(\alpha) (g_\alpha)^{2/\alpha}}{1-2/\alpha} \cdot \frac{\left(\sum_{C}  |C|\left(\frac{\sigma_C}{\smin}\right)^2 \right)^{2/\alpha}\cdot (\smin)^{4/\alpha}\cdot \smax^{2-4/\alpha}}{\sum_{C}  |C|\left(\frac{\sigma_C}{\smin}\right)^2\cdot \smin^2} \cdot  \sum_{i\ge 0} (2^i k_i)^{1-2/\alpha}\\
        &\le \frac{8h(\alpha) (g_\alpha)^{2/\alpha}}{1-2/\alpha} \cdot \left(\frac{\smax}{\smin}\right)^{2-4/\alpha}\cdot \frac{\sum_{i\ge 0} (2^i k_i)^{1-2/\alpha}}{\left(\sum_{C}  |C|\right)^{1-2/\alpha}} \\
        &\le \frac{16h(\alpha) (g_\alpha)^{2/\alpha}}{1-2/\alpha}\cdot \left(\frac{\smax}{\smin}\right)^{2-4/\alpha}\cdot \frac{\sum_{i\ge 0} (2^i k_i)^{1-2/\alpha}}{\left(\sum_{i\ge 0}  2^i k_i\right)^{1-2/\alpha}}\ ,
    \end{align*}
    where the fourth inequality is obtained using
    \begin{equation*}
        \left(\sum_{C}  |C|\left(\frac{\sigma_C}{\smin}\right)^2\right)^{2/\alpha-1} \le \left(\sum_{C}  |C|\right)^{2/\alpha-1}\ .
    \end{equation*}
    To conclude the proof, we need to upper bound the ratio 
    \begin{equation*}
        \frac{\sum_{i\ge 0} (2^i k_i)^{1-2/\alpha}}{\left(\sum_{i\ge 0}  2^i k_i\right)^{1-2/\alpha}}
    \end{equation*}
    for any integer sequence $(k_i)_{i\ge 0}$ satisfying the constraint $\sum_{i\ge 0}k_i=k$. Using Jensen's inequality (note that $x\mapsto x^{1-2/\alpha}$ is concave), we immediately obtain 
    \begin{equation}
    \label{eq:ratio_weights}
        \frac{\sum_{i\ge 0} (2^i k_i)^{1-2/\alpha}}{\left(\sum_{i\ge 0}  2^i k_i\right)^{1-2/\alpha}} \le \ell^{2/\alpha}\ .
    \end{equation}
    For the second upperbound, note that to have equality in Jensen's inequality, we must have that $2^ik_i=2^jk_j$ for all $i,j\ge 0$ such that $k_i \neq 0, k_j\neq 0$. Intuitively, this can only happen when $\ell =O(\log k)$. Formally, let us denote by $L$ the maximum $i\ge 0$ such that $k_L\neq 0$. We then build a decreasing sequence of indices as follows. We define $i_1$ to be equal to $L$. Then, assuming we defined the indices $i_1,i_2,\ldots i_x$, we define $i_{x+1}$ to be the highest index $0\le i<i_x$ such that $k_{i_{x+1}}\ge 2k_{i_x}$. We stop until it is not possible to find an index $i_{x+1}$ fitting those conditions anymore. We denote by $\mathcal I$ the set of indices that were selected, and $\mathcal J:=\mathbb N\setminus \mathcal I$ its complement.

    We notice that, by construction, $|\mathcal I|\le \log k$ since $\sum_{i\in \mathcal I} k_i\le \sum_{i\ge 0} k_i=k$. Then clearly we have that 
    \begin{equation*}
        \frac{\sum_{i\ge 0} (2^ik_i)^{1-2/\alpha}}{\left(\sum_{i\ge 0} 2^ik_i\right)^{1-2/\alpha}} \le \frac{\sum_{i\in  \mathcal I} (2^ik_i)^{1-2/\alpha}}{\left(\sum_{i\in  \mathcal I} 2^ik_i\right)^{1-2/\alpha}} + \frac{\sum_{i\in  \mathcal J} (2^ik_i)^{1-2/\alpha}}{\left(\sum_{i\in  \mathcal I} 2^ik_i\right)^{1-2/\alpha}}\ .
    \end{equation*}
    The first term on the right-hand side can be upperbounded by $|\mathcal I|^{2/\alpha}\le (\log k)^{2/\alpha}$, using Jensen's inequality again. As for the second term, let us denote by $i_1,i_2,\ldots,i_{|\mathcal I|}$ the set of indices selected to be in $\mathcal I$. Then, we write
    \begin{align*}
        \sum_{i\in  \mathcal J} (2^ik_i)^{1-2/\alpha} &= \sum_{x=1}^{|\mathcal I|}\sum_{j\in (i_{x+1},i_{x})} (2^jk_j)^{1-2/\alpha}\\
        &\le \sum_{x=1}^{|\mathcal I|}\sum_{j\in (i_{x+1},i_x)} (2^j 2k_{i_x})^{1-2/\alpha}\\
        &\le \sum_{x=1}^{|\mathcal I|}\sum_{j\in (i_{x+1},i_x)} (2^{i_x}2^{j-i_x} 2k_{i_x})^{1-2/\alpha}\\
        &\le \sum_{i\in \mathcal I} (2^{i}k_{i})^{1-2/\alpha} \cdot \left(\sum_{j=0}^{+\infty} \left(2^{-i}\right)^{1-2/\alpha} \right)\\
        &= \frac{\sum_{i\in \mathcal I} (2^{i}k_{i})^{1-2/\alpha}}{1-2^{2/\alpha-1}}\ .
    \end{align*}
    Therefore the second right-hand side term in Equation \eqref{eq:ratio_weights} can be upper bounded by
    \begin{equation*}
        \frac{(\log k)^{2/\alpha}}{1-2^{2/\alpha-1}}\ ,
    \end{equation*}
    which concludes the proof.
    \end{proof}

\subsection{Wrapping it up}
\label{sec:wrapping_it_up}
Using Lemma \ref{lem:potential_global} in combination with Lemma \ref{lem:potential_upper_bound}, we obtain that the final expected cost of hit clusters is at most 
\begin{equation*}
    h(\alpha) \cdot (g_\alpha)^{2/\alpha} \cdot \left(\frac{\smax}{\smin}\right)^{2-4/\alpha}\cdot  \min \{\ell,\log(k)\}^{2/\alpha}\cdot \OPT\ ,
\end{equation*}
where $h(\alpha):=\frac{16(\alpha/2-1)^{1-2/\alpha}}{\alpha/2-1}\cdot \frac{2-2^{2/\alpha-1}}{1-2^{2/\alpha-1}}$. Regarding the cost of hit clusters, we use Lemma \ref{lem:cost_hit_clusters_alpha_sampling} and Lemma \ref{lem:arthur_cost_hit_uniform_sampling} to see that the expected cost of hit clusters is at most 
\begin{equation*}
    (4e+(\alpha+1)^2\cdot (g_\alpha)^{2/\alpha} )\cdot  \OPT\ .
\end{equation*}

Therefore, if we define $f(\alpha):=(4e+(\alpha+1)^2)\cdot \frac{16(\alpha/2-1)^{1-2/\alpha}}{\alpha/2-1}\cdot \frac{2-2^{2/\alpha-1}}{1-2^{2/\alpha-1}}=O\left(\frac{\alpha^2}{(\alpha/2-1)^{2/\alpha+1}}\right)$, we obtain precisely the result of Theorem \ref{thm:main}.

\section{Additional Experiments}\label{sec:expts}

In this section, we provide additional experiments to further validate our claims. Particularly, we run the $D^{\alpha}$-seeding on the following instances:

\begin{enumerate}
\item 1) $\mathcal{D}_1$: A mixture of $4$ Gaussians with the centers/means of the Gaussians placed on the corners of a square with edge length of $2\Delta=100$ (default). All the covariances are identity matrices (in $d=2$).\
\item 2) $\mathcal{D}_2$: A mixture of $4$ Gaussians with the centers/means of the Gaussians placed on the corners of a square with edge length of $2\Delta=100$ (default). All the covariances are identity matrices (in $d=2$), except one that has a covariance $\sigma^2I$, with $\sigma^2=400$.
\item 3) $\mathcal{D}_3$: A mixture of $8$ Gaussians with the centers/means of the Gaussians placed on the corners of a cube with edge length of $2\Delta=100$ (default). All the covariances are identity matrices (in $d=3$).\
\item 4) $\mathcal{D}_4$: A mixture of $8$ Gaussians with the centers/means of the Gaussians placed on the corners of a cube with edge length of $2\Delta=100$ (default). All the covariances are identity matrices (in $d=3$), except one that has a covariance $\sigma^2I$, with $\sigma^2=800$.\

\item 5) $\mathcal{D}_5$: A mixture of $4$ bivariate student-t distributions, with different degrees of freedom, $\nu = \{1.6,2,5,10\}$, where the location parameter is centered on the corners of a square with edge length of $2\Delta=100$ (default).
\end{enumerate}

We generate these instance with $n=10000$ samples from the appropriate mixture distribution. Then we consider $\alpha=\{2,6,10,\ldots,38\}$, i.e, starting from 2 and increments of 4 for the Gaussian instances, i.e, $\mathcal{D}_1$ through $\mathcal{D}_4$. For $\mathcal{D}_5$ we consider a more fine-grained version of this where $\alpha=\{2,4,6,8,\ldots,16\}$, because the distribution generating $\mathcal{D}_5$ has a heavier tail. So in practice, large $\alpha$ values do not make sense.

In all experiments, we show the average performance over $5000$ trial runs for each value of $\alpha$. In some cases we also run Lloyd's algorithm until convergence after the seeding. 

\begin{figure}
    \centering
    \begin{minipage}{0.5\textwidth}
        \centering
        \includegraphics[width=\textwidth]{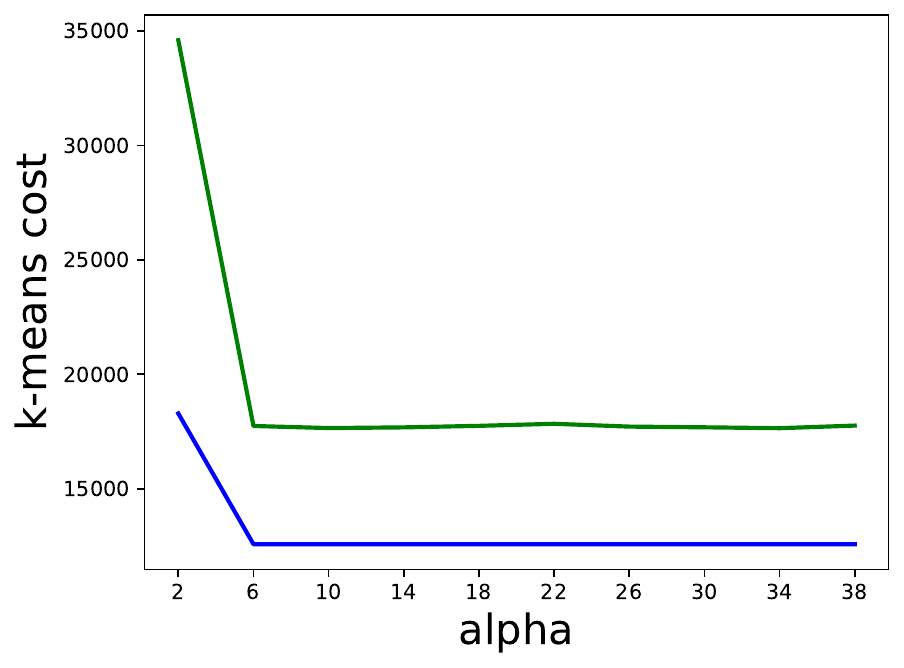} 
    \end{minipage}\hfill
    \begin{minipage}{0.5\textwidth}
        \centering
        \includegraphics[width=\textwidth]{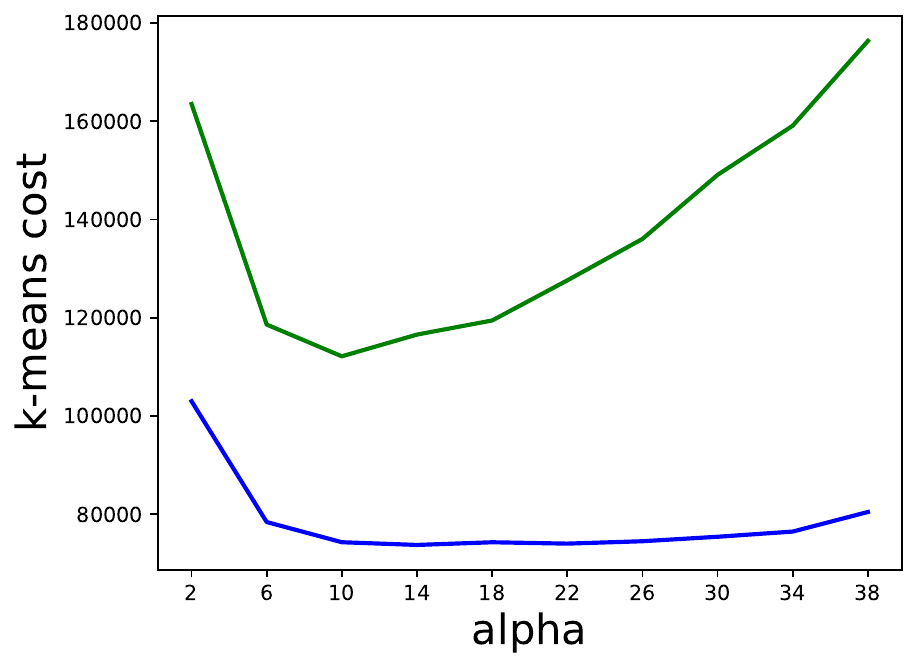} 
    \end{minipage}
    \caption{Performance of $D^{\alpha}$ seeding for two instances $\mathcal{{D}}_1$ (on the left) and $\mathcal{D}_2$ (on the right). The green curve indicates the cost right after the seeding step, and the blue curve indicates the cost if we run Lloyd's algorithm until convergence after the seeding.}
    \label{fig:exptI3vsI4Gaussians}
\end{figure}

\begin{figure}
    \centering
    \begin{minipage}{0.5\textwidth}
        \centering
        \includegraphics[width=\textwidth]{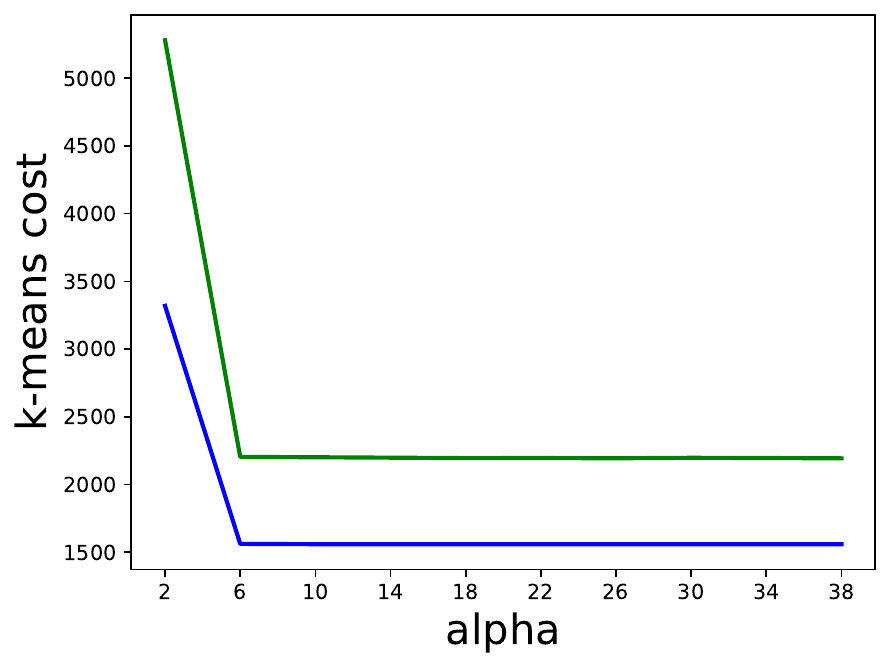} 
    \end{minipage}\hfill
    \begin{minipage}{0.5\textwidth}
        \centering
        \includegraphics[width=\textwidth]{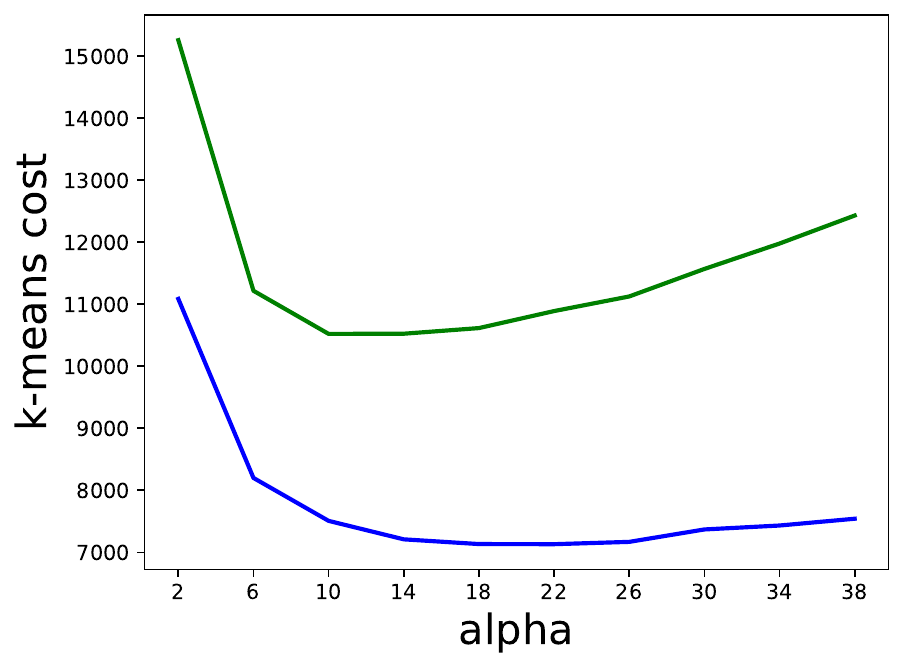} 
    \end{minipage}
    \caption{Performance of $D^{\alpha}$ seeding for two instances $\mathcal{D}_3$ (on the left) and $\mathcal{D}_4$ (on the right). The green curve indicates the cost right after the seeding step and the blue curve indicates the cost if we additionally run Lloyd's steps until convergence. }
    \label{fig:exptI1vsI2Gaussians}
\end{figure}

\begin{figure}
    \centering
    \begin{minipage}{0.5\textwidth}
        \centering
        \includegraphics[width=\textwidth]{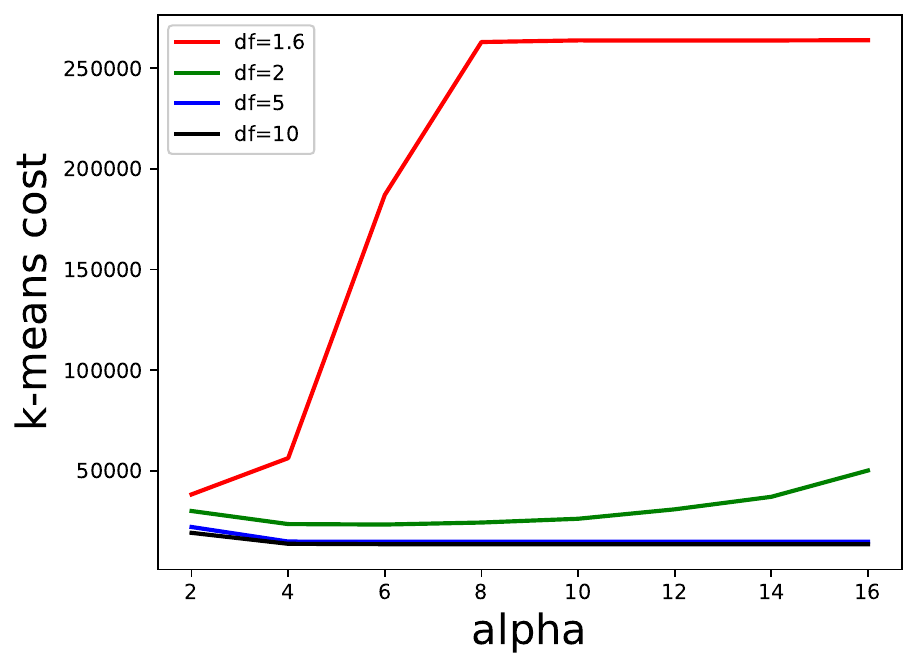} 
    \end{minipage}\hfill
    \begin{minipage}{0.5\textwidth}
        \centering
        \includegraphics[width=\textwidth]{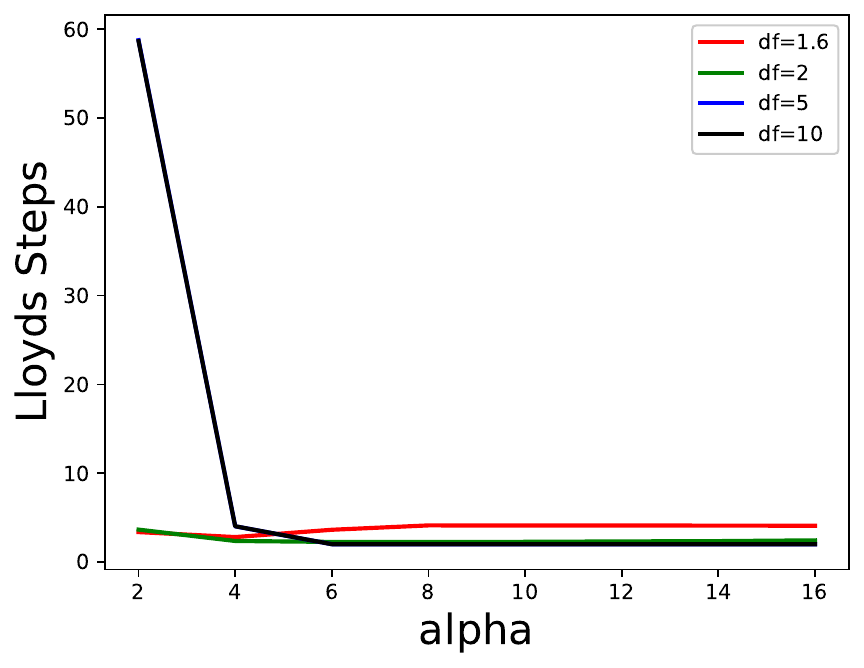} 
    \end{minipage}
    \caption{Performance of $D^{\alpha}$ seeding for the instance $\mathcal{D}_5$. The figure on the left shows the performance, and on the right is the number of Lloyd's steps needed until convergence.}
    \label{fig:studentt}
\end{figure}


 Figures \ref{fig:exptI1vsI2Gaussians} and \ref{fig:exptI3vsI4Gaussians}, reiterate our main claims that for Gaussians of the same/similar covariances, it is always best to choose $\alpha >> 2$. This can be seen in both the $\mathcal{D}_1$  and $\mathcal{D}_3$. However, the story is different when there is one Gaussian with a sufficiently different covariance, then one can see from the figures of $\mathcal{D}_2$ and $\mathcal{D}_4$ that there is a trade-off. Moreover, we can see that even if we run Lloyd's algorithm until convergence the general pattern does not change much. This re-iterates the importance of a \emph{good} initial seeds. 
 
 Finally, with the experiments on student-t distribution, we intend to show the effect of $g_{\alpha}$ on the performance of $D^{\alpha}$ seeding. Here when the degrees of freedom ($df$) are $1< df \leq 2$, the variance is $\infty$ (see \cite{wiki-t}). As seen in Figure \ref{fig:studentt}, when $df=1.6$, $\alpha > 2$ performs poorly in this case and $\alpha = 2$ seems to be relatively better. But when $df=2$, already there are some $\alpha$ values > 2 that are better than $\alpha=2$. In general, one might argue that if such a pattern is observed where $\alpha = 2$ performs better than $\alpha > 2$, it may be that the distributions could have many outliers. In this case, it might be that the $k$-means objective is probably not appropriate and one should consider more robust objectives such as $k$-medians.
 Now with $df=5,10$, where the distributions become more ``Gaussian'' like, we see a similar pattern as the one observed in Gaussian mixtures, where a large $\alpha$ is preferred. Crucially, we also notice that in this case, the number of Lloyd's steps required for convergence it much lower for $\alpha> 2$ in comparison to the number needed when $\alpha=2$.

\section*{Acknowledgements}{This work was partially supported by the Swiss National Science Foundation project 200021-184656 “Randomness in Problem Instances and Randomized Algorithms". The first author is also supported by a Post-Doctoral fellowship of the ETH AI Center, Zurich.}
\bibliographystyle{plain}
\bibliography{biblio}

\appendix
\section{Missing proofs and discussion from Section \ref{sec:thmproof}}
\subsection{On potential functions of previous works}
\label{sec:discussion_potential}
Here, we discuss more in details why a new potential function is needed when dealing with $D^\alpha$ seeding. Many of the previous works on \texttt{k-means++} rely on a more natural potential function. For instance in \cite{sanjoyLN}, the potential function is defined as 
\begin{equation*}
    \psi(t)=\frac{W_t}{|U_t|}\cdot \cost(U_t,Z_t)\ ,
\end{equation*}
where $U_t$ is the number of undiscovered clusters, and $W_t$ is the number of iterations where a center was selected in an already discovered cluster. We notice two main differences with our potential function. The first one is that we partitioned the clusters by weight classes, and the second is that  replaced the expression
\begin{equation*}
    \cost(U^{(i)}_t,Z_t)
\end{equation*} by a more involved 
\begin{equation*}
    (2^i)^{1-2/\alpha}\cdot \sum_{C\in U_t^{(i)}} (\costa(C,Z_t))^{2/\alpha}\ .
\end{equation*}

Let us discuss the second difference first, and for the purpose of simplicity let us assume there is a unique weight class. Then it is very tempting to use the same potential function as \cite{sanjoyLN}. Unfortunately, we run into issues when the next center is selected in an undiscovered cluster. 

Indeed, in that case it is quite crucial in the proof that the potential should decrease in expectation, and an increase in potential here seems like a quite fatal flaw. To show that the potential decreases in the case of $D^2$ seeding, one simply notices that if we had picked one of the undiscovered clusters uniformly at random and removed it from $\cost(U_t,Z_t)$, then the potential would already decrease. Now, $D^2$ seeding can only do better than that, since heavy clusters have even more chance of being sampled. But what about $D^\alpha$ seeding? Unfortunately, this is not true anymore, as it can be that there are two clusters $C_1, C_2$ such that $\cost(C_1,Z_t)>\cost (C_2,Z_t)$, but $\costa(C_1,Z_t)<\costa (C_2,Z_t)$, i.e. the orders are reversed. And this can happen even if $g_\alpha$ is fairly small.

Therefore, one needs a potential where the undiscovered clusters are ranked in the same order of importance as $D^\alpha$ seeding weights them. Of course, this potential should also scale as a square of the euclidean distance. This is where our potential comes in.

Now about the first difference, note that because of scaling issues we had to add an extra $n^{1-2/\alpha}$ if all the clusters are of weight $n$. If the clusters have different weight, the idea of partitioning them using geometric grouping comes naturally.
\subsection{Useful inequalities}
\label{sec:inequalities}
In our proof, we use several times the following inequalities, which we give here for completeness.
\begin{lemma}[Jensen's inequality \cite{hardy1952inequalities}]\label{lem:powermeans}
    For any set of non-negative real numbers $\{y_i\}_{i=1}^{n}$ it holds that:
    \begin{align}
        \sum_{i=1}^{n}y^{\gamma}_i \geq  n^{1-\gamma}\left(\sum_{i=1}^{n}y_i\right)^{\gamma}\ ,
    \end{align}
    if $\gamma\ge 1$. If $0\le \gamma < 1$, the reversed inequality holds.
\end{lemma}

\begin{lemma}[Chebyshev's sum inequality \cite{hardy1952inequalities}]\label{lem:chebyshevsum}
    For any sequence of real numbers $\{y_i\}_{i=1}^{n}$ and $\{w_i\}_{i=1}^{n}$, such that $y_1 \geq y_2 \geq \ldots \geq y_n$ and $w_1 \geq w_2 \geq \ldots \geq w_n$, it holds that:
    \begin{align}
        \frac{1}{n}\sum_{i=1}^{n}y_iw_i\geq  \left(\frac{1}{n}\sum_{i=1}^{n}y_i\right)\cdot\left(\frac{1}{n}\sum_{i=1}^{n}w_i\right).
    \end{align}
\end{lemma}

\section{Lower bounds}
\subsection{The dependence on \texorpdfstring{$\smax/\smin$}{} is tight}\label{app:rmaxoverrmin}
We consider an instance with one cluster of $n$ points arranged on a regular simplex of side-length $R$, while the other $k-1$ clusters contain $n$ points arranged on a regular simplex of side-length $1$. Note that for a cluster made of a regular simplex, we have that $g_\alpha\le 2^\alpha$, hence $(g_\alpha)^{2/\alpha}=O(1)$, and the concentration of clusters will not play any role in this construction.

The clusters of radius $1$ are pairwise separated by a distance $\delta$, while the cluster of radius $R$ is at infinite distance from all the other clusters. We choose $R$ and $\delta$ so that $R^\alpha=10\cdot k\delta^\alpha$. Then, at each iteration $t\ge 3$, the probability of sampling from the cluster of radius $R$ is at least
\begin{equation*}
\frac{nR^\alpha}{nk\delta^\alpha+nR^\alpha}\ge 1/2.
\end{equation*}
Hence, the expected number of missed clusters is at least $k/3$, which implies that the expected cost of the returned solution is at least $(nk/3)\cdot \delta^2$ while $\mbox{OPT}$ costs at most $nR^2+nk$. We select $\delta$ so that $R^2=k$ which is satisfied for $\delta = k^{1/2-1/\alpha}$. In this case, the expected competitive ratio is at least $\Omega (k^{1-2/\alpha})$, and one can check that $\smax = R = k^{1/2}$, and $\smin=1$.

\subsection{The dependence on \texorpdfstring{$g_{\alpha}$}{} is tight}\label{app:galpha}

Consider the following instance. There are two clusters of $n$ points each. The first clusters $C_1$ consists of $n$ points on a regular simplex. All points in $C_1$ are at distance $1$ from each other, and at distance $1$ from the centroid of $C_1$. The second cluster has $2$ groups of points. There are $n-1$ points at the origin, and $1$ point at coordinates $(\Delta,0,\ldots, 0)$. We place the cluster $C_1$ so that its centroid lies at coordinates $(-\delta,0,\ldots ,0)$, with $\delta=\Delta/n^{1/\alpha}$.

Second, we select $\Delta$ so that $\smax/\smin = 1$. For this, we would simply need to set $\Delta$ so that 
\begin{equation*}
    (\Delta/n)^2 \cdot (n-1) + \Delta^2(1-1/n)^2 = n 
\end{equation*}
which gives $\Delta =\Theta(\sqrt{n})$. The cost of the given clustering is equal to $2n$.

However, we can see that $g_{\alpha}$ for the second cluster is not a constant. Indeed we can compute that 
\begin{align}
    g_\alpha&= \frac{(1/n)\cdot \sum_{z\in C_2} \costa(C_1,z)}{n^{1-\alpha/2} \cdot (\cost (C_2,\mu_{C_2}))^{\alpha/2}}\\ 
    &= \frac{(1/n)\cdot (2(n-1)\Delta^\alpha)}{n}\\
    &= \Theta\left(n^{\alpha/2-1}\right)\ .
\end{align}

  To analyze the $D^\alpha$ seeding procedure on this instance, let us denote by $\B_1$ the event that the first sampled center belongs to $C_1$, and by $\B_2$ the event that the second sample belongs to $C_2$ and lies at coordinate $(\Delta,0,\ldots, 0)$.

  Clearly we have that 
  \begin{equation*}
      \mathbb P[\B_1]=1/2\ ,
  \end{equation*}
  and 
  \begin{equation*}
      \mathbb P[\B_2\mid \B_1]\ge \frac{\Delta^\alpha}{n+(n-1)\cdot \delta^\alpha+\Delta^\alpha}\ge 1/3\ ,
  \end{equation*}
  by our choice of $\delta$. This implies that $\mathbb P[\B_1\cap \B_2]\ge 1/6$ and the expected cost of the output clustering is at least 
  \begin{equation*}
      \frac{1}{6}\cdot (n-1)\cdot \delta^2 = \Theta(n^{1+1-2/\alpha}) = \Omega(n^{1-2/\alpha}\cdot \OPT) = \Omega((g_\alpha)^{2/\alpha}) \cdot \OPT\ .
  \end{equation*}




\subsection{A lower bound for greedy \texorpdfstring{$k$}-means++}
\label{sec:scikit}
In this section, we prove Theorem \ref{thm:greedylowerbound}. In particular, it implies a super constant lower bound for the greedy variant which uses $f(k)$ samples, as long as $f(k)$ is super constant. We emphasize that here we did not try to optimize the lower bound. The main purpose of this section is to show that the greedy algorithm with a superconstant number of samples cannot guarantee a constant factor in expectation. Consider the following construction against greedy $k$-means++ with $f(k)$ samples. We will place our points in the euclidean space. For ease of construction, we will take the distributional point of view in this instance, i.e. we assume that each cluster $C_i$ contains infinitely many points placed in $\mathbb R^d$ according to some distribution $f_i$. Moreover, we will assume that these clusters are balanced (i.e. they have the same weight) which means here that the global distribution of all points in all the clusters has probability density $f(x)=\frac{1}{k}\cdot \sum_{i\in [k]} f_i(x)$ for all $x\in \mathbb R^2$.

In our instance, we will create $k/4$ groups of $4$ clusters as follows. First, we describe how to construct one group. We create an square $ABCD$ of side length $a=\log(f(k))$ and we denote by $M$ the centroid of this square, and by $b=a/\sqrt{2}$ the length of half a diagonal in the square. In the cluster $C_1$, the probability density will be a one-dimensional density on the segment $[AM]$ defined as follows
\begin{equation*}
    f_1(x)=
    \begin{cases}
    \frac{e^{-||x-A||_2}}{1-e^{-b}} \mbox{ if $x\in [AM]$ and}\\
    0 \mbox{ otherwise.}
    \end{cases}
\end{equation*}
We recall here that $AM=b$ so the above function is indeed a probability density. The centroid of the first cluster will be equal to the point $\mu_1\in [AM]$ such that
\begin{equation*}
    A\mu_1=\frac{\int_{0}^{b}xe^{-x}dx}{1-e^{-b}} = 1-\frac{b}{e^{b}-1}\ .
\end{equation*}

We repeat the construction similarly for the clusters $C_2$, $C_3$, and $C_4$, replacing $A$ by $B$, $C$, and $D$ respectively. Note that we needed only $2$ dimensions to construct our first group. We refer the reader to Figure \ref{fig:lowerbound_greedy} for an illustation of the construction.

\begin{figure}
    \centering
    \includegraphics{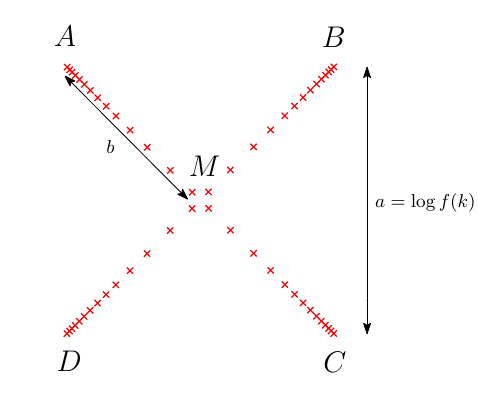}
    \caption{A schematic view of a group of $4$ clusters. The densest parts are at the vertices $A,B,C,D$. All the points in a cluster lie on the segment from a vertex to the center of mass of the square.}
    \label{fig:lowerbound_greedy}
\end{figure}

Then, we add $k/4-1$ extra dimensions and we fix the point $M$ to be the origin of $\mathbb R^{k/4+1}$. Finally, we make $k/4-1$ additional copies $T_2,T_3,\ldots, T_{k/4}$ of our first square and we arrange their centroids $M_1=(0)_{k/4+1},M_2,\ldots ,M_{k/4}$ in a regular simplex of side-length $\Delta=100\cdot (f(k))^3\cdot k$. One should think of $\Delta$ as a much bigger distance than $a$ (essentially, we have $k/4$ squares that are infinitely far from each other). Each group of $4$ cluster will be denoted $G_j$ for $j\in [k/4]$.

We prove the following simple statements.
\begin{claim}
\label{cla:Laplace_cost}
For any cluster $C_i$, we have that 
\begin{align*}
   \lim_{a\to \infty} \cost (C_i)=1. 
  \end{align*}
\end{claim}
\begin{proof}
    We can write
    \begin{align*}
        \cost (C_i) &= \frac{\int_{0}^b (x-A\mu_i)^2e^{-x}dx}{1-e^{-b}}\\
        &=\frac{e^b+e^{-b}-b^2-2}{(e^b-1)\cdot (1-e^{-b})}\ ,
    \end{align*}
    which clearly tends to $1$ as $b$ goes to infinity.
\end{proof}
\begin{claim}
\label{cla:LaplaceMoments}
For any cluster $C_i$ and any fixed $\alpha\ge 2$ and $b>10$, we have that 
\begin{align}
     g_\alpha= \frac{\int_{0}^b\int_0^b |x-y|^{\alpha}f_i(x)f_i(y)dxdy}{\left(\int_0^b (x-A\mu_i)^{2}f_i(x)dx\right)^{\alpha/2}}\le 4\cdot \Gamma(\alpha+1)\ ,
  \end{align}
  where $\Gamma$ is the gamma function \cite{artin2015gamma}.
\end{claim}
\begin{proof}
From Claim \ref{cla:Laplace_cost}, we already have that the denominator tends to $1$ as $a$ goes to infinity. Hence, we only need to upper bound the numerator. We proceed with the simple change of variable $u=(x+y)$ and $v=(x-y)$. Then we obtain
\begin{align*}
    \int_{0}^b\int_0^b |x-y|^{\alpha}f_i(x)f_i(y)dxdy &= \frac{1}{2(1-e^{-b})^2}\cdot \int_{-b}^{b}\int_{|v|}^{2b-|v|} |v|^{\alpha}e^{-u} dudv \\
    &= \frac{1}{(1-e^{-b})^2}\cdot \int_{0}^{b}\int_{v}^{2b-v} v^{\alpha}e^{-u} dudv \\
    &=\frac{1}{(1-e^{-b})^2}\cdot \int_{0}^{b} v^{\alpha}(e^{-v}-e^{v-2b}) dudv \\
    &\le 2\int_{0}^{\infty}v^{\alpha}e^{-v} dv\\
    &= 2\Gamma(\alpha+1)\ .
\end{align*}
Using the fact that we consider $b\ge 10$, we obtain the desired upperbound.
\end{proof}

In light of Claim \ref{cla:LaplaceMoments}, it is is clear that for any fixed $\alpha>2$, the $D^\alpha$ seeding algorithm will guarantee a constant factor approximation in expectation.

However, we use the instances we build to prove Theorem \ref{thm:greedylowerbound}. To this end, we will need to look at the cost of a group of $4$ clusters after the greedy $k$-means++ algorithm selected exactly one center inside a group $j$. We distinguish two key cases: (1) If the first center $x$ in the cluster $C_i\in G_j$ such that the distance between $x$ and the centroid $\mu_i$ of the cluster is less than $\log (b)$, and (2) if the first center is at distance at most $b/100$ from the centroid $M_j$ of the corresponding square. Let us denote by $K_1$ and $K_2$ the total cost of the $4$ clusters in $G_j$ after case (1) or (2) happened respectively.

\begin{claim}
    \label{cla:Laplace_cost_cases}
    For $k$ a big enough constant, we have that 
    \begin{equation*}
        K_1\ge 11b^2/2\ ,
    \end{equation*}
    and 
    \begin{equation*}
        K_2\le 5b^2\ .
    \end{equation*}
\end{claim}
\begin{proof}
    Assume w.l.o.g. that the first center belongs to center $C_1$. In the first case then the total cost is lower bounded by the cost of the other three clusters $C_2,C_3,C_4$. Assume $C_3$ is the center whose centroid $\mu_3$ is the furthest away from $\mu_1$. Then the cost of this center is lower bounded by
    \begin{equation*}
        \int_{0}^b (2b-\log(b)-x)^2e^{-x}dx \ge 4b^2 + O(b\log b)\ .
    \end{equation*}
    For the other two clusters $C_2$ and $C_4$, we use Pythagore theorem. Since they lie on a perpendicular line to $C_1$ and we are in case (1), the distance of any point in $C_2$ or $C_4$ is at distance at least $b-\log b$ from the sampled center. Hence we can lower bound the cost of each of the other two clusters with the following quantity
    \begin{equation*}
        \int_{0}^b (b-\log(b))^2f(x)dx \ge b^2+O(b\log b)\ .
    \end{equation*}
    Hence the total cost is at least 
    \begin{equation*}
        6b^2+O(b\log b)
    \end{equation*}
    which is more than $11b^2/2$ for $b$ big enough.

    In the second case, the cost of each cluster is upperbounded by 
    \begin{equation*}
        \int_{0}^b (1.01b-x)^2f(x)dx \le 1.03b^2+O(b)\ .
    \end{equation*}
    Hence the total cost is upperbounded by 
    \begin{equation*}
        4.12b^2+O(b)\ ,
    \end{equation*}
    which is less than $5b^2$ for $b$ big enough.
\end{proof}

Next, we define as $\mathcal L$ the event that the greedy $k$-means++ algorithm samples one of its $f(k)$ samples at distance at most $b/100$ from the centroid $M_j$ of some still undiscovered group $G_j$.
\begin{claim}
\label{cla:proba_bad_event}
    Assume we are at iteration $t<k/4$. Then 
    \begin{equation*}
        \mathbb P[\mathcal L]\ge 1/2.
    \end{equation*}
\end{claim}
\begin{proof}
    At iteration $t<k/4$, it must be that one group $G_j$ is still undiscovered. Then we have that (by triangle inequality) 
    \begin{equation*}
        \cost(G_j)\ge 4(\Delta-2a)^2 > \Delta^2\ ,
    \end{equation*} for $k$ a big enough constant. The total cost of the discovered clusters is at most
    \begin{equation*}
        k\cdot (2a)^2 \le k (2\log f(k))^2 \le \cost(G_j)/(10f(k))\ .
    \end{equation*}
    By union-bound we obtain that with probability at least $9/10$, none of the sampled candidate centers belong to an already discovered group. Conditionned on that event, we can simply compute by triangle inequality that the probability that a specific sample is at distance at most $b/100$ from the centroid is at least (assuming $b$ and $f(k)$ are some big enough constant)
    \begin{equation*}
        \frac{9}{10}\cdot \left(\int_{99b/100}^b \frac{e^{-x}}{1-e^{-b}}dx \right)\ge \frac{9}{10}\cdot (e^{-99b/100}-e^{-b})\ge \frac{9}{10}e^{-99b/100}\ge \frac{9}{10(f(k))^{1/\sqrt{2}}}\ge \frac{1}{f(k)}\ .
    \end{equation*}
    Therefore, with probability at most $(1-1/f(k))^{f(k)}\le 1/e$, we have that none of the sampled candidate centers are close the centroid $M_j$ of a square. Therefore with probability $1/2$ at least one of them is close to a centroid $M_j$ of an undiscovered group. 
\end{proof}

When the event $\mathcal{L}$ happens, it must be (by Claim \ref{cla:Laplace_cost_cases}) that the selected center is at distance at least $\log b$ from the center it belongs to. By Claim \ref{cla:proba_bad_event}, this happens in expectation at least $k/8$ times. This means that the expected cost of the output clustering is at least
\begin{equation*}
    (k/8) \cdot \int_{0}^{\log b} (x-\log b)^2 e^{-x}dx \ge (k/8) \cdot ((\log b)^2+O(\log b)) = \Omega ((k/8)\cdot (\log \log (f(k)))^2\ .
\end{equation*}
Since by Claim \ref{cla:Laplace_cost} we have that the optimum cost is equivalent to $k$ (for $k$ growing to infinity), this concludes the proof of Theorem \ref{thm:greedylowerbound}.


\end{document}